\def\isarxiv{1}

\def\paperTitle{Sublinear Time Quantum Algorithm for Attention Approximation}

\def\paperAuthor{
Zhao Song\thanks{\texttt{magic.linuxkde@gmail.com}. Simons Institute for the Theory of Computing, UC Berkeley.}
\and
Jianfei Xue\thanks{\texttt{jx898@nyu.edu}. New York University.}
\and
Jiahao Zhang\thanks{\texttt{ml.jiahaozhang02@gmail.com}.}
\and
Lichen Zhang\thanks{\texttt{lichenz@mit.edu}. Massachusetts Institute of Technology.}
}

\ifdefined\isarxiv
\documentclass[11pt]{article}
\usepackage[numbers]{natbib}
\else
\documentclass{article}
\usepackage{iclr2026_conference, times}
\iclrfinalcopy
\fi

\ifdefined\isarxiv

\usepackage{amsmath}
\usepackage{amsthm}
\usepackage{amssymb}
\usepackage{algorithm}
\usepackage{subfig}
\usepackage{algpseudocode}
\usepackage{graphicx}
\usepackage{grffile}
\usepackage{wrapfig,epsfig}
\usepackage{url}
\usepackage{xcolor}
\usepackage{epstopdf}

\usepackage{bbm}
\usepackage{dsfont}

\else 

\usepackage[utf8]{inputenc} 
\usepackage[T1]{fontenc}    
\usepackage{hyperref}       
\usepackage{url}            
\usepackage{booktabs}       
\usepackage{amsfonts}       
\usepackage{nicefrac}       
\usepackage{microtype}      
\usepackage{xcolor}         

\fi

\ifdefined\isarxiv

\else
\usepackage{amsmath}
\usepackage{amsthm}
\usepackage{amssymb}
\usepackage{algorithm}
\usepackage{subfig}
\usepackage{algpseudocode}
\usepackage{graphicx}
\usepackage{grffile}
\usepackage{wrapfig,epsfig}
\usepackage{epstopdf}
\usepackage{bbm}
\usepackage{dsfont}
\fi
 
\allowdisplaybreaks

\ifdefined\isarxiv

\usepackage{tikz}
\usepackage{hyperref}  
\hypersetup{colorlinks=true,citecolor=blue,linkcolor=blue} 
\usetikzlibrary{arrows}
\usepackage[margin=1in]{geometry}

\else
\definecolor{mydarkblue}{rgb}{0,0.08,0.45}
\hypersetup{colorlinks=true, citecolor=mydarkblue,linkcolor=mydarkblue}
\fi
 
\graphicspath{{./figs/}}

\theoremstyle{plain}
\newtheorem{theorem}{Theorem}[section]
\newtheorem{lemma}[theorem]{Lemma}
\newtheorem{definition}[theorem]{Definition}

\newtheorem{corollary}[theorem]{Corollary}

\newtheorem{remark}[theorem]{Remark}

\newcommand{\wh}{\widehat}
\newcommand{\wt}{\widetilde}
\newcommand{\ov}{\overline}

\newcommand{\R}{\mathbb{R}}

\DeclareMathOperator*{\E}{{\mathbb{E}}}

\DeclareMathOperator{\poly}{poly}

\DeclareMathOperator{\tr}{tr}

\DeclareMathOperator{\Att}{Att}

\begin{document}

\ifdefined\isarxiv

\date{}
\title{\paperTitle}
\author{\paperAuthor}

\else

\title{\paperTitle}

\author{%
  Zhao Song \\
  Simons Institute for the Theory of Computing, UC Berkeley \\
  \texttt{magic.linuxkde@gmail.com}
  \And
  Jianfei Xue \\
  New York University \\
  \texttt{jx898@nyu.edu} \\
  \And 
  Jiahao Zhang \\
  \texttt{ml.jiahaozhang02@gmail.com} \\
  \And
  Lichen Zhang \\
  MIT CSAIL \\
  \texttt{lichenz@csail.mit.edu}
}

\maketitle

\fi

\ifdefined\isarxiv
\begin{titlepage}
  \maketitle
  \begin{abstract}

Given the query, key and value matrices $Q, K, V\in \mathbb{R}^{n\times d}$, the attention module is defined as $\mathrm{Att}(Q, K, V)=D^{-1}AV$ where $A=\exp(QK^\top/\sqrt{d})$ with $\exp(\cdot)$ applied entrywise, $D=\mathrm{diag}(A{\bf 1}_n)$. The attention module is the backbone of modern transformers and large language models, but explicitly forming the softmax matrix $D^{-1}A$ incurs $\Omega(n^2)$ time, motivating numerous approximation schemes that reduce runtime to $\widetilde O(nd)$ via sparsity or low-rank factorization.

We propose a quantum data structure that approximates any row of $\mathrm{Att}(Q, K, V)$ using only row queries to $Q, K, V$. Our algorithm preprocesses these matrices in
\[
\widetilde{O}\left( \epsilon^{-1} n^{0.5} \left( s_\lambda^{3} + s_\lambda^{2} d + d \right) + \epsilon^{-2} s_\lambda d \right)
\]
time, where $\epsilon$ is the target accuracy and $s_\lambda$ is the $\lambda$-statistical dimension of the exponential kernel defined by $Q$ and $K$. Each row query can be answered in
$\widetilde{O}(s_\lambda^2 + s_\lambda d)$
time. If every row sum of $A$ exceeds $\epsilon\|A\|+\lambda\sqrt n$, then with probability $1-1/\mathrm{poly}(n)$ the answers have Frobenius-norm error $O(\beta\|D^{-1}\|(\epsilon\|A\|_F+\lambda\sqrt n)\|V\|_F)$. Here $\beta=(1-(\epsilon\|A\|+\lambda\sqrt n)\|D^{-1}\|)^{-1}\geq 1$ is close to $1$ when every row sum is much larger than $\epsilon\|A\|+\lambda\sqrt n$.

To our knowledge, this is the first quantum data structure that approximates rows of the attention matrix in sublinear time with respect to $n$. Our approach relies on a quantum Nystr{\"o}m approximation of the exponential kernel, quantum multivariate mean estimation for computing $D$, and quantum sampling of the rows of $V$ by their squared norms, with a median over independent samples, for the multiplication with $V$.

  \end{abstract}
  \thispagestyle{empty}
\end{titlepage}

\newpage

\else

\begin{abstract}

Given the query, key and value matrices $Q, K, V\in \mathbb{R}^{n\times d}$, the attention module is defined as $\mathrm{Att}(Q, K, V)=D^{-1}AV$ where $A=\exp(QK^\top/\sqrt{d})$ with $\exp(\cdot)$ applied entrywise, $D=\mathrm{diag}(A{\bf 1}_n)$. The attention module is the backbone of modern transformers and large language models, but explicitly forming the softmax matrix $D^{-1}A$ incurs $\Omega(n^2)$ time, motivating numerous approximation schemes that reduce runtime to $\widetilde O(nd)$ via sparsity or low-rank factorization.

We propose a quantum data structure that approximates any row of $\mathrm{Att}(Q, K, V)$ using only row queries to $Q, K, V$. Our algorithm preprocesses these matrices in
\[
\widetilde{O}\left( \epsilon^{-1} n^{0.5} \left( s_\lambda^{3} + s_\lambda^{2} d + d \right) + \epsilon^{-2} s_\lambda d \right)
\]
time, where $\epsilon$ is the target accuracy and $s_\lambda$ is the $\lambda$-statistical dimension of the exponential kernel defined by $Q$ and $K$. Each row query can be answered in
$\widetilde{O}(s_\lambda^2 + s_\lambda d)$
time. If every row sum of $A$ exceeds $\epsilon\|A\|+\lambda\sqrt n$, then with probability $1-1/\mathrm{poly}(n)$ the answers have Frobenius-norm error $O(\beta\|D^{-1}\|(\epsilon\|A\|_F+\lambda\sqrt n)\|V\|_F)$. Here $\beta=(1-(\epsilon\|A\|+\lambda\sqrt n)\|D^{-1}\|)^{-1}\geq 1$ is close to $1$ when every row sum is much larger than $\epsilon\|A\|+\lambda\sqrt n$.

To our knowledge, this is the first quantum data structure that approximates rows of the attention matrix in sublinear time with respect to $n$. Our approach relies on a quantum Nystr{\"o}m approximation of the exponential kernel, quantum multivariate mean estimation for computing $D$, and quantum sampling of the rows of $V$ by their squared norms, with a median over independent samples, for the multiplication with $V$.

\end{abstract}

\fi



\section{Introduction}

Transformers~\citep{vsp+17} have emerged as one of the most successful machine learning architectures in recent years, revolutionizing fields such as natural language processing~\citep{dclt19,ydy+19,rsr+20,bmr+20,jys+20}, computer vision~\citep{cms+20,dbk+21,gxl+22}, speech recognition~\citep{cbs+15,wdcx21}, robotics~\citep{lew+22}, and time series forecasting~\citep{zzp+21}. These models typically operate on sequences of length $n$, autoregressively predicting the next most likely token to produce an output of length $n$. In applications like large language models (LLMs), it has been widely observed that increasing the sequence length $n$ significantly enhances generative performance. However, this benefit comes at a substantial computational cost: the core attention module has a quadratic time complexity in $n$, which severely limits both training and inference scalability.

Formally, let $Q, K, V \in \mathbb{R}^{n \times d}$ denote the query, key, and value embeddings. The attention module is defined as
\begin{align*}
    \Att(Q, K, V) = D^{-1} A V \in \mathbb{R}^{n \times d},
\end{align*}
where
$
A = \exp(QK^\top / \sqrt{d}) \in \mathbb{R}^{n \times n}$
is computed entrywise, and
$D = \mathrm{diag}(A {\bf 1}_n) \in \mathbb{R}^{n \times n}$. The matrix $A$ is referred to as the \emph{attention matrix}, and $D^{-1}A$ as the \emph{softmax matrix}. Due to the $n \times n$ size of $A$, much recent research has focused on reducing the quadratic complexity by approximating attention through pattern-based sparse attention~\citep{dkod20,kkl20,rsvg21,syy22,cgrs19,bpc20,aoa+20,zgd+20}, linearizing the kernel through feature mapping~\citep{kvpf20,cld+21,wlk+20,ppy+21}, or various algorithmic and data structure optimizations~\citep{zhdk23,as23,hjk+24,kmz24,zhmk24,bsz24,syz24,kbkw25,cam+25,cls+25,iksw25}.

The theoretical goal in these efforts is to achieve a runtime that scales nearly linearly with $n$, allowing some approximation error. This is a natural target, since the input size to the attention module is $n \times d$. On a classical computer, any algorithm that approximates attention in time $\widetilde{O}(nd)$ is considered optimal. But could this process be accelerated further using a quantum computer?

If our objective is to output the entire $n \times d$ matrix $\Att(Q, K, V)$, then $\Omega(nd)$ time is unavoidable due to output size. However, in many transformer applications — particularly during inference~\citep{pdc+23,bmn+24,aaj+24,zwl+24,flc+25,lro+24,kwz+24,bfz+25,clk+25,cll+25,iksw25} — only \emph{row queries} are needed. In this setting, we aim to preprocess $Q, K, V$ into a data structure such that, for any index $i \in [n]$, the structure can return a vector $\widetilde{r}_i \in \mathbb{R}^d$ that approximates the $i$-th row of $\Att(Q, K, V)$. This model circumvents the $\Omega(nd)$ lower bound by focusing on partial output. Nonetheless, since each row of $\Att(Q, K, V)$ is a convex combination of rows of $V$, achieving truly sublinear time in $n$ still appears classically intractable.

In this work, we answer this question affirmatively. Specifically, we construct a quantum data structure that preprocesses $Q, K, V$ using only \emph{row queries}, and does so in time\footnote{We use $\widetilde{O}(\cdot)$ to suppress polylogarithmic factors in $n$, $d$, $s_\lambda$, and $1/\epsilon$.} $\widetilde{O}(\epsilon^{-1} n^{0.5} \cdot \mathrm{poly}(d, s_\lambda)+\epsilon^{-2}s_\lambda d)$,
where $s_\lambda$ is the \emph{statistical dimension} of the exponential kernel matrix associated with $Q$ and $K$. Given any index $i \in [n]$, the data structure returns an approximation to the $i$-th row of $\Att(Q, K, V)$ in time $\widetilde{O}(s_\lambda^2 + s_\lambda d)$.

To our knowledge, this is the first quantum algorithm to implement the row query model in sublinear time. Prior works either require superlinear preprocessing time or impose structural assumptions~\citep{gsyz23}. Our approach avoids both: it makes \emph{no structural assumptions} on $Q, K, V$ (it needs only a lower bound on the row sums of $A$ and a nondegeneracy condition on the keys and the sampled points, see Theorem~\ref{thm:main_formal}), making it broadly applicable in practice. Moreover, our construction is conceptually simple — it combines quantum techniques such as Grover search~\citep{g96}, Nystr{\"o}m kernel approximation, quantum multivariate mean estimation~\citep{chj22} and amplitude estimation~\citep{bhmt02} to approximate each component of the attention module: $D$, $A$, and $V$.

\paragraph{Quantum Computation Model.} We follow the standard quantum computation framework as in~\cite{aw22,ag23}. The model allows quantum subroutines using $O(\log n)$ qubits, quantum queries to the input, and access to a quantum-read/classical-write RAM (QRAM) of $\mathrm{poly}(n)$ bits. Each quantum read or classical write takes unit cost. We measure \emph{time complexity} by the number of QRAM operations, and \emph{query complexity} by the number of queries to the input. In our setting, we query rows of $Q$, $K$, and $V$, each requiring $O(d)$ time classically. For simplicity, we assume $Q$ and $K$ have been scaled by $1/d^{1/4}$, which can also be done via row queries in $O(d)$ time.

\section{Preliminary}

\paragraph{Notation.} Given symmetric matrices $A, B\in \R^{n\times n}$, we use $A-B\succeq 0$ to denote $A-B$ is a positive semidefinite (PSD) matrix, i.e., for any $x\in \R^n$, $x^\top (A-B)x\geq 0$. Given a matrix $M\in \R^{n\times n}$, we use $\exp(M)$ to denote the entrywise exponentiation operation. We use $\tr[M]$ to denote the trace of $M$. For a real matrix $A$, we use $A^\dagger$ to denote its Moore-Penrose pseudoinverse, and for a square, nonsingular real matrix $M$, we use $M^{-1}$ to denote its inverse. For two vectors $x, y\in \R^n$, we use $x^\top y$ or $\langle x,y\rangle$ to denote the inner product of $x$ and $y$. We use ${\bf 0}_n$ and ${\bf 1}_n$ to denote all-0's and all-1's vector. For a vector $x\in \R^n$, we use $\|x\|_2=\sqrt{x^\top x}$ to denote its $\ell_2$ norm, $\|x\|_\infty=\max_{i\in[n]}|x_i|$ to denote its $\ell_\infty$ norm. If $M$ is a PSD matrix, then we use $\|x\|_M=\sqrt{x^\top Mx}$ to denote the $M$-energy norm of $x$. For a matrix $A$, we use $\|A\|$ to denote its spectral norm and $\|A\|_\infty$ to denote its max row $\ell_1$ norm, and $\|A\|_F$ to denote its Frobenius norm. Throughout the paper, we will also exclusively work with weighted sampling matrices, usually denoted by $S\in \R^{n\times s}$ for where $s$ is the total number of samples taken, let $i(j)$ be the index of the $i$-th sample, then the $i$-th column of $S$ is $\frac{1}{\sqrt{p_j}}e_j$, where $p_j$ is the probability of choosing the index $j$. We use $\E[X]$ to denote the expectation of a random variable $X$. We use $\mathbb{I}[E]$ to denote the indicator of whether event $E$ happens.

\paragraph{Numerical Linear Algebra.} We rely on several primitives from numerical linear algebra for fast approximations and provable guarantees.

\begin{definition}[Leverage score]
\label{def:ls}
Let $A \in \mathbb{R}^{n \times d}$. The $i$-th leverage score of $A$ is defined as
\begin{align*}
    \tau_i := a_i^\top (A^\top A)^{-1} a_i,
\end{align*}
where $a_i$ is the $i$-th row of $A$. Equivalently, let $A=U\Sigma V^\top$ be its SVD, then $\tau_i=\|u_i\|_2^2$, where $u_i$ is the $i$-th row of $U$.
\end{definition}

We will also work exclusively with \emph{kernel matrices}. Given a dataset $X = \{x_1, \ldots, x_n\} \subseteq \mathbb{R}^d$, we define the exponential kernel matrix $E \in \mathbb{R}^{n \times n}$ by $E_{i,j} = \exp(\langle x_i, x_j \rangle)$. Although $E$ is generally full-rank, our algorithm depends only on a parameter called the \emph{$\lambda$-statistical dimension} of $E$, which may be much smaller than $n$.

\begin{definition}[Statistical dimension~\citep{z05,htf09}]
\label{def:stat_dim}
Let $E \in \mathbb{R}^{n \times n}$ be a PSD matrix, and let $\lambda > 0$. The $\lambda$-statistical dimension of $E$ is defined as $s_\lambda(E) := \mathrm{tr}[E(E + \lambda I)^{-1}]$. When $E$ is clear from context, we write $s_\lambda$ for simplicity.
\end{definition}

Note that $s_\lambda$ is a monotonically decreasing function of $\lambda$, and is closely related to the notion of ridge leverage scores.

\begin{definition}[Ridge leverage score~\citep{am15}]
\label{def:ridge_ls}
Let $E \in \mathbb{R}^{n \times n}$ be a kernel matrix and let $\lambda > 0$. The $\lambda$-ridge leverage score of the data point $x_i$ is defined as
\begin{align*}
    \tau_i^\lambda := (E(E + \lambda I)^{-1})_{i,i}.
\end{align*}
If $E = BB^\top$ for some $B \in \mathbb{R}^{n \times n}$, then this can be equivalently written as
\begin{align*}
    \tau_i^\lambda = b_i^\top (B^\top B + \lambda I)^{-1} b_i,
\end{align*}
where $b_i$ is the $i$-th row of $B$.
\end{definition}

It is easy to see that $\sum_{i=1}^n \tau_i^\lambda = s_\lambda$. Moreover,~\cite{mm17} shows that Nystr{\"o}m approximations~\citep{ws00} based on ridge leverage score sampling yield accurate spectral approximations to $E$.

\begin{lemma}[Theorem 3 of~\cite{mm17}]
\label{lem:ridge_ls_mm17}
Let $s = O(s_\lambda \log(s_\lambda / \delta))$, $\lambda > 0$, and $\delta \in (0, 1)$. Let $E \in \mathbb{R}^{n \times n}$ be any kernel matrix. Let $S \in \mathbb{R}^{n \times s}$ be the $\lambda$-ridge leverage score sampling matrix. Then the Nystr{\"o}m approximation $\widetilde{E} := E S (S^\top E S)^\dagger S^\top E$ satisfies $\widetilde{E} \preceq E \preceq \widetilde{E} + \lambda I$
with probability at least $1 - \delta$.
\end{lemma}

\paragraph{Quantum Primitives.} In this paper, we primarily leverage two quantum algorithmic primitives. The first is an efficient quantum sampling oracle based on Grover search.

\begin{lemma}[Claim 3 in~\cite{aw22}]
\label{lem:q_sample}
Let $n$ be a positive integer, and let $\{p_1, \ldots, p_n\} \subseteq [0,1]$ be a list of probabilities, each given with $\widetilde{O}(1)$ bits of precision. There exists a quantum algorithm, $\textsc{QSample}(p)$, that generates a list of indices where each $i$ is sampled independently with probability $p_i$, in expected time $\widetilde{O}\left( \sqrt{n \sum_{i=1}^n p_i} \right) \cdot {\cal T},$
where ${\cal T}$ denotes the time required to generate any individual $p_i$.
\end{lemma}

\begin{remark}\label{rem:q_sample_whp}
The time bound of Lemma~\ref{lem:q_sample} holds in expectation, not in the worst case, and the lemma assumes probabilities with $\widetilde O(1)$ bits. Our algorithms compute real numbers $p_i$ and state running times that hold with high probability, so we run $\textsc{QSample}$ as follows, using the proof of Claim~3 in~\cite{aw22}. Let $b:=\lceil\log_2 n\rceil$ and $p_i':=2^{-b}\lceil 2^bp_i\rceil$. Then $p_i\leq p_i'\leq \min\{1,p_i+1/n\}$, and $p_i'$ is described by $b+1$ bits. The proof of Claim~3, applied to the $p_i'$, reads $b+1$ bits of a uniformly random string for each index $i$, which marks $i$ with probability $p_i'$, independently of the other indices, and it finds the set $S'$ of marked indices by repeated Grover search. By Claim~2 of~\cite{aw22}, this search uses $\widetilde O(\sqrt{n\max\{1,|S'|\}})$ elementary operations and evaluations of the $p_i'$, and it returns $S'$ with probability at least $2/3$. Hence the guarantee for the list that $\textsc{QSample}(p')$ returns is that it equals $S'$ with probability at least $2/3$, and we use Lemma~\ref{lem:q_sample} only through the following procedure. For a parameter $\eta\in (0,1)$, we run the search $r:=\lceil\log_3(2/\eta)\rceil$ times on the same marks. One evaluation of $p_i'$ and one read of the $b+1$ bits of $i$ tell whether a returned index $i$ is marked, so we discard every returned index that is not marked, and we take the union of the results. This union is contained in $S'$, and it equals $S'$ except with probability at most $3^{-r}\leq \eta/2$. Finally, we keep each index $i$ of the union independently with probability $p_i/p_i'$. Let $\widehat S$ be the set of indices of $S'$ that this last step keeps when it is applied to all of $S'$. Every index belongs to $\widehat S$ independently with probability $p_i'\cdot p_i/p_i'=p_i$, and the output equals $\widehat S$ whenever the union equals $S'$. We call $\widehat S$ the exact sample of the run.

The size $|S'|$ is a sum of independent indicator variables with mean $\sum_i p_i'\leq \sum_i p_i+1$. The Chernoff bound $\Pr[X\geq x]\leq 2^{-x}$ for $x\geq 6\E[X]$ shows that, with probability at least $1-\eta/2$, we have $|S'|\leq \max\{6(\sum_i p_i+1),\log_2(2/\eta)\}$. Since $|S'|\leq n$, on this event the run, including the checks and the evaluations of $p_i$ in the last step, takes $\widetilde O(r\sqrt{n\max\{\sum_i p_i+1,\log(2/\eta)\}})\cdot {\cal T}$ time. Apers and de Wolf bound the running time with a Chernoff bound in the same way in the proof of their Theorem~11. Hence, with a uniformly random string, with probability at least $1-\eta$ the run takes this time and returns its exact sample.

Writing down a uniformly random string of $(b+1)n$ bits takes linear time. The proof of Claim~3 removes the string with Corollary~1 of~\cite{aw22}. This corollary simulates an algorithm with running time $q$ that queries a uniformly random string by an algorithm with running time $\widetilde O(q)$ and $\widetilde O(q)$ additional bits of QRAM that queries a $2q$-wise independent string, built by Theorem~7 of~\cite{aw22}, and by Fact~1 of~\cite{aw22} the two algorithms have the same output distribution. The bound $q$ must be fixed before the run. So, for an integer $m\geq 1$, let $q(m)=\widetilde O(r\sqrt{nm})\cdot {\cal T}$ be a bound on the time of the procedure above when $|S'|\leq m$, and let a run with budget $m$ be this procedure, stopped once its time exceeds $q(m)$ and simulated in this way with a fresh string. Every event about the output of such a run has the same probability as with a uniformly random string, so we analyze the run with such a string. If $M\geq \sum_i p_i'$ and $m\geq \max\{6M,\log_2(2/\eta)\}$, then $|S'|\leq m$ except with probability at most $2^{-m}\leq \eta/2$ by the Chernoff bound above, and in this case the run finishes within its budget. Since $3^{-r}\leq \eta/2$, the run with budget $m$ returns its exact sample except with probability at most $\eta$. Hence, conditioned on everything computed before the run, its output is within total variation distance $\eta$ of an independent sample, by which we mean a set that contains each index $i$ independently with probability $p_i$.

In Algorithm~\ref{alg:q_nystrom}, the sum $\sum_i p_i$ depends on $s_\lambda$, which is never computed, so we find such a bound $M$ by doubling. Let $\mu:=\sum_i p_i'$, $J:=\lceil\log_2(24n/\eta)\rceil$ and $\eta_1:=\eta/(2(J+1))$. Each run below is a run of the procedure above with $\eta_1$ in place of $\eta$, so that $r=\lceil\log_3(2/\eta_1)\rceil=O(\log(n/\eta))$, and it uses fresh marks and a fresh string. For $j=0,1,\ldots,J$ in turn, we make a run with budget $m_j:=2^j$, and we stop at the first $j$ at which this run finishes and its union has at most $m_j/2$ elements. We set $M:=\max\{m_j,8\ln(2/\eta_1)\}$ for this $j$, or $M:=n$ if there is no such $j$. Finally, we make a run with budget $\lceil\max\{6M,\log_2(2/\eta_1)\}\rceil$ and return its output. We now show that $\mu\leq M$ except with probability at most $\eta/2$. Since $\mu\leq n$, this fails only if we stop at a level $j$ with $\mu>\max\{m_j,8\ln(2/\eta_1)\}$. At such a level, we stop only if the union differs from the set $S'$ of the run, which has probability at most $\eta_1/2$, or if $|S'|\leq m_j/2<\mu/2$. By the Chernoff bound $\Pr[X\leq \E[X]/2]\leq e^{-\E[X]/8}$, the latter has probability at most $e^{-\mu/8}<\eta_1/2$. A union bound over the $J+1$ levels proves the claim. Next, let $j^*$ be the least $j$ with $m_j/2\geq \max\{6\mu,\log_2(2/\eta_1)\}$. Then $j^*\leq J$, because $2^{J-1}\geq 12n/\eta$ is at least $6n\geq 6\mu$ and at least $\log_2(2/\eta_1)$. If we reach the level $j^*$, then we stop there whenever $|S'|\leq m_{j^*}/2$, since the run then finishes within its budget and its union is contained in $S'$. Hence we stop at a level $j\leq j^*$ except with probability at most $2^{-m_{j^*}/2}\leq \eta_1/2$. On this event, $m_j\leq m_{j^*}<4\max\{6\mu,\log_2(2/\eta_1)\}$, so $M=O(\sum_i p_i+\log(n/\eta))$, and since the budgets grow geometrically, all runs together take $\widetilde O(r\sqrt{n\max\{\sum_i p_i+1,\log(n/\eta)\}})\cdot {\cal T}$ time. If $\mu\leq M$, then the output of the last run is within total variation distance $\eta_1$ of an independent sample. So, conditioned on everything computed before $\textsc{QSample}$ starts, its output is within total variation distance $\eta/2+\eta_1$ of an independent sample, and we couple the two, using randomness that nothing else uses, so that they agree except with probability at most $\eta/2+\eta_1$. From now on, the exact sample of a run of $\textsc{QSample}$ is this coupled set. Because the run and the coupling use fresh randomness, conditioned on everything computed before the run and on the exact samples of earlier runs, it contains each index $i$ independently with probability $p_i$. Since $\eta/2+\eta_1+\eta_1/2\leq \eta$, with probability at least $1-\eta$ the run takes $\widetilde O(r\sqrt{n\max\{\sum_i p_i+1,\log(n/\eta)\}})\cdot {\cal T}$ time and returns its exact sample. For $\eta=n^{-c}$ with a constant $c\geq 1$, we have $r=O(\log n)$, so this time is $\widetilde O(\sqrt{n\max\{1,\sum_i p_i\}})\cdot {\cal T}$. The algorithm does not need to know $\sum_i p_i$.
\end{remark}

The second primitive is a quantum procedure for approximating matrix-vector products using quantum multivariate mean estimation.

\begin{lemma}[Theorem 5.1 of~\cite{ag23}]
\label{lem:q_matvec}
Let $\epsilon \in (0,1)$, and let $A \in \mathbb{R}^{n \times d}$ have full column rank and $v \in \mathbb{R}^n$. Suppose we are given quantum query access to the rows of $A$ and the entries of $v$. Then there exists a quantum algorithm $\textsc{QMatVec}(A, v, \epsilon)$ that outputs a vector $\widetilde{\mu} \in \mathbb{R}^d$ such that, with probability at least $1 - 1/\mathrm{poly}(n)$,
$
    \| \widetilde{\mu} - A^\top v \|_{(A^\top A)^{-1}} \leq \epsilon,
$
using
$
    \widetilde{O}\left( \epsilon^{-1} n^{0.5} d \|v\|_\infty \right)
$
row queries to $A$ and queries to entries of $v$, and $\widetilde{O}(\epsilon^{-1} n^{0.5} d^3 \|v\|_\infty+d^\omega)$ time, counting each row query as $O(d)$ time.
\end{lemma}

The dependence on $d$ is linear because the algorithm applies quantum multivariate mean estimation~\citep{chj22} to a $d$-dimensional random variable whose covariance matrix $\Sigma$ has trace $O(nd\|v\|_\infty^2)$. That estimator uses $\widetilde{O}(\sqrt{d\,\tr[\Sigma]}/\epsilon)=\widetilde{O}(\epsilon^{-1}n^{0.5}d\|v\|_\infty)$ samples. Section~\ref{sec:leverage} also uses quantum maximum finding~\citep{dh96} and amplitude estimation~\citep{bhmt02}.

\section{Technical Overview}
In this section, we provide an overview on the algorithmic techniques we utilize to approximate $A, D$ and $V$, in sublinear time.

\subsection{Approximate the Attention Matrix via Quantum Nystr{\"o}m}

To approximate the attention matrix $A$, we will make use of Nystr{\"o}m approximation~\citep{ws00}. However, recall that $A = \exp(QK^\top)$; for $Q \neq K$, the matrix itself is not even symmetric. This poses significant challenges for obtaining a good approximation. On the other hand, if we treat the queries and keys as the \emph{dataset}, and form the exponential kernel matrix over them, then the resulting matrix is indeed a kernel matrix.

Specifically, let the dataset $X = \{q_1, \ldots, q_n, k_1, \ldots, k_n\}$, and consider $E \in \mathbb{R}^{2n \times 2n}$ where
$E = \begin{bmatrix}
    \exp(QQ^\top) & \exp(QK^\top) \\
    \exp(KQ^\top) & \exp(KK^\top)
\end{bmatrix}$,
then the attention matrix can be retrieved via $P E \begin{bmatrix} {\bf 0}_{n\times n} \\ I_n \end{bmatrix}$ where $P \in \mathbb{R}^{n \times 2n}$ is the matrix consisting of the first $n$ rows of the $2n \times 2n$ identity matrix, which selects the first $n$ rows of $E$. Thus, once we obtain an approximation for $E$, we automatically obtain an approximation for $A$.

It remains to compute a Nystr{\"o}m approximation of $E$, as at first glance it is not clear how to even generate the ridge leverage score sampling matrix $S$ in sublinear time. ~\cite{mm17} shows that on a classical computer, it is possible to compute a \emph{generalized} ridge leverage score sampling matrix using $\widetilde{O}(n s_\lambda)$ evaluations of the kernel function and an additional $\widetilde{O}(n s_\lambda^2)$ time, via a recursive sampling scheme:
\begin{itemize}
    \item Uniformly sample half of the data points, then recursively compute the weighted sampling matrix $\widetilde{S}^{n \times s}$ for the subset;
    \item Compute the \emph{generalized ridge leverage score}, defined as $\widetilde{\tau}_i^\lambda := b_i^\top (B^\top \widetilde{S} \widetilde{S}^\top B + \lambda I)^\dagger b_i$, and set $p_i = \min\{1, \widetilde{\tau}_i^\lambda \cdot \log(\sum_{j} \widetilde{\tau}_j^\lambda / \delta)\}$;
    \item Output $S$ as the weighted sampling matrix according to $p_i$.
\end{itemize}

The key ingredients in their algorithm are (1) the generalized ridge leverage score can be computed via kernel function evaluations instead of computing the factorization (see Definition~\ref{def:gen_ridge_ls}), and (2) sampling according to generalized ridge leverage score only increases the sample size by a constant factor, hence it does not affect the asymptotic runtime of the algorithm (see Lemma~\ref{lem:classical_nystrom}).

For the simpler setting of leverage score sampling,~\cite{ag23} shows that this recursive framework can benefit from quantum speedup, especially the Grover search sampler of Lemma~\ref{lem:q_sample}, by noting that when sampling according to the leverage score, it is not necessary to compute or approximate all the scores; rather, it is enough to implement an oracle that can supply any approximate leverage score when needed.

For our application, however, this oracle is much more difficult to implement, as in the setting of~\cite{ag23}, one could directly query the row of $B$, which is not the case for the kernel setting. Nevertheless, we show how to implement such an oracle for generalized ridge leverage scores of kernels. The oracle supplies one score per call, while the factor $\log(\sum_{j} \widetilde{\tau}_j^\lambda / \delta)$ of~\cite{mm17} depends on all the scores, and computing it would take a linear number of oracle calls. We replace it by the fixed factor $\log(8n^2/\delta)$. The guarantee of~\cite{mm17} still holds, and the sample size is $O(s_\lambda\log(n/\delta))$ (Theorem~\ref{thm:q_nystrom}). The algorithm is detailed in Algorithm~\ref{alg:q_nystrom}. Throughout this section, we let $s$ denote the final sample size of the Nystr{\"o}m approximation.

\begin{algorithm}[!ht]
\caption{Quantum Nystr{\"o}m approximation via recursive generalized ridge leverage score sampling.}
\label{alg:q_nystrom}
\begin{algorithmic}[1]
\Procedure{QNystr{\" o}mKernel}{$\{x_1, \ldots, x_n\} \in (\mathbb{R}^d)^n, \mathsf{K} : \mathbb{R}^d \times \mathbb{R}^d \rightarrow \mathbb{R}, \delta\in (0, 1),\lambda\in (0,\infty)$} \Comment{$\delta$ is the failure probability, $\lambda$ is the ridge leverage score parameter.}
    \State \Comment{$s_\lambda$ is never computed. Below, $s$ denotes the current number of samples, which is $O(s_\lambda\log(n/\delta))$ with probability at least $1-\delta$}
    \State $T \gets \lceil\log_2(n/m_0)\rceil$, where $m_0=O(\log(n/\delta))$ bounds the expected size of the base case $S_0$
    \State Let $S_0 \subset_{1/2} S_1 \subset_{1/2} \cdots \subset_{1/2} S_T = [n]$ \Comment{We use $A\subset_{1/2} B$ to denote that $A$ keeps each index of $B$ independently with probability $1/2$}
    \State $M_0 \gets \{\mathsf{K}(x_i, x_j)\}_{(i,j) \in S_0 \times S_0}$ \Comment{$|S_0| = O(\log(n/\delta))$ with probability at least $1-\delta$}
    \State Let $D_0 \in \mathbb{R}^{n \times |S_0|}$ be the sampling matrix of $S_0$
    \For{$t = 1$ to $T$}
        \State $\widehat{M} \gets (M_{t-1} + \lambda I)^{-1}$
        \State \Comment{Let $D_{t-1}^\top K_i := \{D_{t-1}(j) \cdot \mathsf{K}(x_i, x_j)\}_{j \in D_{t-1}} \in \mathbb{R}^s$ for $i \in S_t$ where $D_{t-1}(j)$ is the weight corresponding to $x_j$ specified by $D_{t - 1}$}
        \State Implement oracle for $q_i \gets \frac{5}{\lambda} \cdot ( \mathsf{K}(x_i, x_i) - (D_{t-1}^\top K_i)^\top \widehat{M} D_{t-1}^\top K_i )$ for $i \in S_t$
        \State \Comment{$p_i = \min\{1, 16q_i \log(8n^2/\delta)\}$, the logarithmic factor is fixed}
        \State $\widetilde{D}_t \gets \textsc{QSample}(p)$ \Comment{Run as in Remark~\ref{rem:q_sample_whp}, $\widetilde{D}_t \in \mathbb{R}^{|S_t| \times s}$}
        \State $D_t \gets D_{S_t} \cdot \widetilde{D}_t$ \Comment{$D_t \in \mathbb{R}^{n \times s}$}
        \State $M_t \gets \{D_t(i) D_t(j) \cdot \mathsf{K}(x_i, x_j)\}_{(i,j) \in D_t \times D_t}$ \Comment{$M_t \in \mathbb{R}^{s \times s}$}
    \EndFor
    \State \Return $D_T$
\EndProcedure
\end{algorithmic}
\end{algorithm}
The main idea is to utilize the identity $\widetilde{\tau}_i^\lambda = \frac{1}{\lambda}(E - ES (S^\top ES + \lambda I)^{-1} S^\top E)_{i,i}$, where $E_{i,i}$ involves a single kernel evaluation ${\sf K}(x_i, x_i)$, and $S^\top ES$ requires only $O(s^2)$ kernel evaluations. Finally, the term $(ES(S^\top ES + \lambda I)^\dagger S^\top E)_{i,i}$ can be computed by evaluating the kernel between $x_i$ and the sampled points in $S$, weighted appropriately, which requires $O(s)$ kernel evaluations. This shows that we can implement the oracle by precomputing $(S^\top ES + \lambda I)^\dagger$ in $O(s^2) \cdot {\cal T}_{\sf K} + s^\omega$ time, where ${\cal T}_{\sf K}$ denotes the time for kernel evaluation and $\omega$ is the matrix multiplication exponent, for which the current best bound is $\omega<2.258$ over every field of characteristic zero~\cite[Theorem~1.1 and Corollaries~14.1--14.2]{OpenAIMM26}. 
Each oracle query can then be answered in $O(s) \cdot {\cal T}_{\sf K} + s^2$ time, and each probability $p_i$ needs one oracle query because the logarithmic factor is fixed. By Lemma~\ref{lem:q_sample} and Remark~\ref{rem:q_sample_whp}, with probability at least $1-\delta$ the quantum sampler requires only $\widetilde{O}(n^{0.5} s^{0.5}\log(n/\delta))$ oracle calls, so the overall runtime is $\widetilde{O}(n^{0.5} s^{1.5}\log(n/\delta) \cdot ({\cal T}_{\sf K} + s) + s^\omega)$. In our setting, the kernel function ${\sf K}(x_i, x_j) = \exp(\langle x_i, x_j \rangle)$ can be computed in $O(d)$ time, which gives a runtime of $\widetilde{O}(n^{0.5} s^{1.5}\log(n/\delta)(d + s) + s^\omega)$, sublinear in $n$.

It remains to analyze the approximation guarantee. Sampling according to generalized ridge leverage scores ensures that $\widetilde{E} \preceq E \preceq \widetilde{E} + \lambda I$, but this does not immediately imply a bound on the approximation error for $\exp(QK^\top)$. To address this, let 
$E = \begin{bmatrix} B & A \\ A^\top & C \end{bmatrix}$ and $\widetilde{E} = \begin{bmatrix} \widetilde{B} & \widetilde{A} \\ \widetilde{A}^\top & \widetilde{C} \end{bmatrix}$. Standard spectral approximation theory guarantees that $\widetilde{B} \preceq B \preceq \widetilde{B} + \lambda I$ and $\widetilde{C} \preceq C \preceq \widetilde{C} + \lambda I$. For the off-diagonal block we are interested in $A$, we cannot get such a strong spectral approximation guarantee; in fact, one can show that the best we could hope for is a symmetrization bound: $-\lambda I\preceq (\wt A+\wt A^\top)-(A+A^\top)\preceq \lambda I$. On the other hand, a weaker and a more handy bound can be exhibited: $\|A-\wt A\|\leq \lambda$ and $\|A-\wt A\|_F\leq \lambda\sqrt{n}$, and we will show these bounds are sufficient to derive the final approximation guarantees of our algorithm.

It is also worth noting that Algorithm~\ref{alg:q_nystrom} merely computes the weighted sampling matrix $S$, which can be stored compactly by recording the sampled indices and corresponding weights, but does not explicitly form the Nystr{\"o}m approximation $\widetilde{E} = ES (S^\top ES)^\dagger S^\top E$. While $(S^\top ES)^\dagger$ can be computed and stored in $O(s^2 d + s^\omega)$ time, forming $\widetilde{E}$ would take $\Omega(ns)$ time, which is prohibitive due to output size. In what follows, we show that this restricted representation of $S$ is nonetheless sufficient to approximate $D$, $V$, and $\Att(Q, K, V)$.

We now compare our Nystr{\"o}m approximation scheme to a related method known as Nystr{\"o}m-former~\citep{xzc+21}, which also integrates Nystr{\"o}m into the attention mechanism. Specifically, they consider the attention matrix $A$ and partition it as 
$A = \begin{bmatrix} X_1 & X_2 \\ X_3 & X_4 \end{bmatrix}$, aiming to approximate $X_4$ using the other three blocks. Given Nystr{\"o}m landmark points $Q'$ and $K'$ sampled from $Q$ and $K$, they set $X_1 = \exp(Q'K'^\top)$, $X_2 = \exp(QK'^\top)$, and $X_3 = \exp(Q'K^\top)$. Since the number of landmarks is small, these blocks are all low-dimensional.~\cite{xzc+21} proves that $X_4$ can be efficiently approximated using $X_1$, $X_2$, and $X_3$ in $O(nmd)$ time, where $m$ is the number of landmarks. While Nystr{\"o}mformer performs well in practice, it guarantees convergence to the true attention matrix only when all rows of $Q$ and $K$ are included as landmarks. In contrast, our Nystr{\"o}m scheme operates on the exponential kernel matrix formed from $Q$ and $K$, and achieves spectral approximation guarantees as long as the sample size is sufficiently large without needing to include all data points.

\subsection{Approximate the Normalization Factor via Quantum Mean Estimation}

Recall that $D = {\rm diag}(A {\bf 1}_n)$, and each normalization factor only requires computing $a_i^\top {\bf 1}_n$, where $a_i$ is the $i$-th row of $A$. If we have access to $\widetilde{E}$, then the $i$-th normalization factor could be estimated as $\widetilde{E}_{i,*}^\top \begin{bmatrix} {\bf 0}_n \\ {\bf 1}_n \end{bmatrix}$. However, as discussed earlier, we cannot explicitly form $\widetilde{E}$ due to its size. To resolve this, we define $U := ES (S^\top ES)^{\dagger/2} \in \mathbb{R}^{2n \times s}$. By the definition of the Nystr{\"o}m approximation, we have $\widetilde{E} = UU^\top$. 
Moreover, $U$ also exhibits a block structure $U=\begin{bmatrix}
    U_1 \\
    U_2
\end{bmatrix}$ where $U_1, U_2\in \R^{n\times s}$, and the desired approximate $\wt A=U_1U_2^\top$ can be obtained via these blocks. Given any vector $v\in \R^n$, if we can compute or approximate $U_2^\top v$, then the normalization factor for the $i$-th row can be estimated as $(U_1)_{i,*}^\top (U_2^\top v)$ where $(U_1)_{i,*}\in \R^s$ is the $i$-th row of $U_1$. Fortunately, we can implement row queries to $U_2$. We first precompute $(S^\top ES)^{\dagger/2}$ in $O(s^2 d + s^\omega)$ time, then each row $(U_2)_{i,*}$ of $U_2$ is computed via kernel evaluations between $x_{i+n}$ and the points in $S$, followed by matrix-vector multiplication with $(S^\top ES)^{\dagger/2}$. This takes $O(s^2 + s d)$ time. 

It remains to approximate $U_2^\top v$, which we cast as a multivariate mean estimation problem. Define the random variable $X = n v_i (U_2)_{i,*}$, where $i \in [n]$ is selected uniformly at random. It is easy to verify that $\mathbb{E}[X] = U_2^\top v$, and the variance is bounded. Therefore, one can apply the quantum multivariate mean estimation procedure of~\cite{chj22} to approximate $U_2^\top v$. To further reduce variance,~\cite{ag23} proposes approximating the matrix-vector product in the $(U_2^\top U_2)^{-1}$-energy norm. Following this idea, and assuming that $U_2$ has full column rank as Lemma~\ref{lem:q_matvec} requires, we apply Lemma~\ref{lem:q_matvec} to output a vector $\widetilde{\mu} \in \mathbb{R}^s$ such that
$\| \widetilde{\mu} - U_2^\top v \|_{(U_2^\top U_2)^{-1}} \leq \epsilon$,
using $\widetilde{O}(\epsilon^{-1} n^{0.5} s \|v\|_\infty)$ row queries to $U_2$ and $v$. In our application, we always have $\|v\|_\infty = 1$, and as noted above, each row query to $U$ takes $O(s^2 + s d)$ time. We present the full algorithm below in Algorithm~\ref{alg:norm_est}.

\begin{algorithm}[!ht]
\caption{Algorithm for estimating normalization factor.}
\label{alg:norm_est}
\begin{algorithmic}[1]
\State {\bf data structure} \textsc{QRowNorm}
\State {\bf begin members}
\State \hspace{4mm} $s\in \mathbb{N}$
\State \hspace{4mm} $S\in (\R^2)^s$
\State \hspace{4mm} $N\in \R^{s\times s}$
\State \hspace{4mm} $\wt \mu \in \R^s$
\State {\bf end members}
\State
\Procedure{Preprocess}{$Q\in \R^{n\times d}, K\in \R^{n\times d}, \lambda\in (0,\infty),\epsilon\in (0,1)$}
\State $S\gets \textsc{QNystr{\"o}mKernel}(Q\cup K, (x_i,x_j)\mapsto \exp(\langle x_i, x_j\rangle), 1/\poly(n), \lambda)$ \Comment{Algorithm~\ref{alg:q_nystrom}, $S$ is a list of sampled indices and weights}
\State $s\gets$ the number of columns of $S$ \Comment{$s=O(s_\lambda\log n)$ with high probability}
\State $N\gets (S^\top ES)^{\dagger/2}$
\State Implement row oracle $(U_2)_{j,*}$ as follows:
\State \hspace{4mm} $\wt (U_2)_{j(k),*}\gets S_{k}\cdot \exp(\langle x_{j+n},x_k\rangle), \forall k\in S$ \Comment{$\wt (U_2)_{j(k),*}\in \R^s$} 
\State \hspace{4mm} $(U_2)_{j,*}\gets N \wt (U_2)_{j(k),*}$ \Comment{$S$ stores pairs of indices and weights, $S_k$ is the weight corresponding to index $k$, $(U_2)_{j,*}\in \R^s$}
\State Implement entry oracle for a vector $v={\bf 1}_n\in \R^{n}$
\State $\wt \mu\gets \textsc{QMatVec}(U_2, v, \epsilon)$ \Comment{$\wt \mu\in \R^s$, Lemma~\ref{lem:q_matvec}}
\EndProcedure
\State
\Procedure{Query}{$i\in [n]$}
\State $b_i\gets \langle (U_1)_{i,*},\wt \mu\rangle$ \Comment{$(U_1)_{i,*}$ is computed via row oracle}
\State \Return $b_i$
\EndProcedure
\State {\bf end data structure}
\end{algorithmic}
\end{algorithm}

For the approximation guarantee, we prove that for any vector $x \in \mathbb{R}^s$, if $\|x\|_{(U_2^\top U_2)^{-1}} \leq \epsilon$, then $\|U_1 x\|_2\leq \epsilon\cdot \|U_1U_2^\top\|$. This is particularly useful for us, as we can set $x = U_2^\top v - \widetilde{\mu}$, in which case $U_1x=U_1U_2^\top v-U_1 \wt \mu=\wt Av-U_1 \wt \mu$, and the upper bound becomes $\epsilon\cdot \|\wt A\|\leq \epsilon\cdot (\|A\|+\lambda)$. On the other hand, we can upper bound $\|(\widetilde{A} - A)v\|_\infty$ using the matrix infinity norm, defined as $\|\widetilde{A} - A\|_\infty = \max_{i \in [n]} \|\widetilde{A}_{i,*} - A_{i,*}\|_1$. A simple argument shows that $\|\widetilde{A} - A\|_\infty \leq \sqrt{n} \cdot \|\widetilde{A} - A\| \leq \lambda \sqrt{n}$. A triangle inequality then yields the final approximation guarantee. If we define $\widetilde{D} := {\rm diag}(U_1\widetilde{\mu})$, the above analysis provides a bound on $\|D - \widetilde{D}\|$. However, in forming the attention module, it is more desirable to control $\|\widetilde{D}^{-1}\|$. To achieve this, we prove a perturbation bound on matrix inversion that relates $\|\widetilde{D}^{-1}\|$ to $\|D^{-1}\|$.

\subsection{Approximate the Value Matrix via Row Norm Sampling}

In preceding discussions, we have shown how to construct the sampling matrix for Nystr{\"o}m approximation and how to compute the normalization factor for any row $i \in [n]$. It remains to approximate $V$ in sublinear time. Prior classical algorithms, such as~\cite{zhdk23}, propose using importance sampling based on the \emph{joint row norm} of $V$ and $D^{-1}A$. Specifically, the sampling probability for the $i$-th row is set as $p_i \geq 1/4 \cdot (\|D^{-1}Ae_i\|_2^2 + \gamma \cdot \|v_i\|_2^2)/(\|D^{-1}A\|_F^2 + \gamma \cdot \|V\|_F^2)$, where $\gamma = \|D^{-1}A\|^2/\|V\|^2$. This method achieves a final sample size that is nearly linear in $d + {\rm srank}(D^{-1}A)$, where ${\rm srank}(D^{-1}A) = \|D^{-1}A\|_F^2/\|D^{-1}A\|^2$ is the stable rank of the softmax matrix. While this approach is conceptually simple and easy to implement, it requires the column norms of $D^{-1}A$ as well as the Frobenius norms of both matrices, and we do not know how to estimate the column norms of $D^{-1}A$ in sublinear time. We instead sample the rows of $V$ by their squared norms alone~\citep{dkm06}. Each row $i$ is kept independently with probability $p_i=\min\{1,s_0\|v_i\|_2^2/Z\}$, where $Z$ is within a factor of $2$ of $\|V\|_F^2$. Quantum maximum finding~\citep{dh96} and amplitude estimation~\citep{bhmt02} compute such a $Z$ with $\widetilde{O}(n^{0.5})$ row queries. Then $\textsc{QSample}$ (Lemma~\ref{lem:q_sample} and Remark~\ref{rem:q_sample_whp}) draws the sample with $\widetilde{O}((ns_0)^{0.5})$ row queries with high probability, since $\sum_i p_i\leq s_0$. 

Unlike the joint sampling distribution of~\cite{zhdk23}, which yields a \emph{spectral norm approximate matrix multiplication} guarantee of the form $\|D^{-1}ASS^\top V-D^{-1}AV\| \leq \epsilon \cdot \|D^{-1}A\| \cdot \|V\|$, sampling by the rows of $V$ alone gives a guarantee in Frobenius norm, $\|WSS^\top V-WV\|_F\leq \epsilon\|W\|_F\|V\|_F$ for any fixed $W$. This guarantee holds only with constant probability when $s_0=O(\epsilon^{-2})$ (Lemma~\ref{lem:amm}). Sampling by the leverage scores of $V$ gives the same guarantee with $s_0$ multiplied by the \emph{row distortion} $\alpha := d/\|V\|_F^2 \cdot \max_{i:\,\tau_i>0} \|v_i\|_2^2/\tau_i$ of $V$ (Definition~\ref{def:row_distortion}). Intuitively, $\alpha$ measures the mismatch between the row density $\|v_i\|_2^2/\|V\|_F^2$ and the row importance $\tau_i/d$. We show that $1\leq \alpha\leq d/{\rm srank}(V)$, and $\alpha=1$ if $V$ has orthonormal columns, which recovers the result of~\cite{cw17}. Leverage scores also cost $\widetilde{O}((nd)^{0.5})$ row queries to estimate~\citep{ag23}, so squared row norms are the better choice here.

The constant success probability cannot be raised to $1-1/\poly(n)$ with a single squared-norm or leverage-score sample of $\widetilde{O}(\epsilon^{-2})$ rows (Remark~\ref{rem:amm_tight}). We therefore draw $T=O(\log n)$ independent samples $S_V^{(1)},\ldots,S_V^{(T)}$ during preprocessing. A query for row $i$ forms the $T$ candidate rows $e_i^\top \wt D^{-1}\wt A S_V^{(t)}S_V^{(t)\top}V$ and returns the candidate whose distance to its $\lceil T/2\rceil$-th nearest candidate (counting itself) is smallest. Since $\wt D^{-1}\wt A$ does not depend on the samples of $V$, with probability $1-1/\poly(n)$ more than half of the candidates are accurate for every row. The returned candidate is then accurate as well (Lemma~\ref{lem:amm_median}). In total, the $T$ samples cost $\widetilde{O}(\epsilon^{-1} n^{0.5})$ row queries to $V$ and $\widetilde{O}(\epsilon^{-1} n^{0.5} d)$ time, and forming the $T$ products with the sampled rows takes $\widetilde{O}(\epsilon^{-2} s_\lambda d+s_\lambda^2 d)$ further time.

\subsection{Main Result}
Now that we have described how to approximate each of the matrices $D$, $A$, and $V$, we are in a position to state our main result. We provide an overview of our algorithm below in Algorithm~\ref{alg:q_attention}.

\begin{algorithm}[!ht]
\caption{Quantum data structure for attention row query.}
\label{alg:q_attention}
\begin{algorithmic}[1]
\State {\bf data structure} \textsc{QAttention} \Comment{Theorem~\ref{thm:main_informal}}
\State {\bf begin members}
\State \hspace{4mm} $s_E, T\in \mathbb{N}$
\State \hspace{4mm} $\wt L^{(1)},\ldots,\wt L^{(T)}\in \R^{s_E\times d}$
\State \hspace{4mm} $\textsc{QRowNorm}$ \text{QRN} \Comment{Algorithm~\ref{alg:norm_est}}
\State {\bf end members}
\State 
\Procedure{Preprocess}{$Q\in \R^{n\times d}, K\in \R^{n\times d}, V\in \R^{n\times d}, \lambda>0, \epsilon\in (0,1)$}
\State \Comment{$s_\lambda$ is never computed, because Algorithm~\ref{alg:q_nystrom} uses a base case of expected size at most $m_0=O(\log(n/\delta))$ and $\lceil\log_2(n/m_0)\rceil$ levels, and its sample sizes adapt to $O(s_\lambda\log(n/\delta))$}
\State $s_0\gets \lceil 72/\epsilon^2\rceil$, $T\gets \lceil 8\ln(n/\delta)\rceil$ with $\delta=1/\poly(n)$ \Comment{Theorem~\ref{thm:q_amm}, Lemma~\ref{lem:amm_median}}
\State $\text{QRN}.\textsc{Preprocess}(Q, K, \lambda,\epsilon/2)$ \Comment{Algorithm~\ref{alg:norm_est}}
\State $Z\gets \textsc{QFrobNorm}(V)$ \Comment{Lemma~\ref{lem:q_frob}, $\|V\|_F^2\leq Z\leq 2\|V\|_F^2$ with high probability}
\State \Comment{If $Z=0$, which happens when $V=0$, skip the loop below and let $\textsc{Query}$ return ${\bf 0}_d$}
\State $S_E\gets \textsc{QNystr{\"o}mKernel}(Q\cup K, (x_i,x_j)\mapsto \exp(\langle x_i, x_j\rangle), 1/\poly(n), \lambda)$ 
\State $s_E\gets$ the number of columns of $S_E$ \Comment{$s_E=\wt O(s_\lambda)$ with high probability}
\State \Comment{Let $x_1,\ldots,x_{2n}$ denote the dataset $Q\cup K$}
\State $\wt M\gets \{S_E(i) S_E(j)\cdot \exp(\langle x_i, x_j\rangle)  \}_{(i,j)\in S_E\times S_E}$ \Comment{$\wt M\in \R^{s_E\times s_E}$}
\For{$t=1$ to $T$}
\State $S_V^{(t)}\gets \textsc{QSample}(p)$ with $p_j=\min\{1,s_0\|v_j\|_2^2/Z\}$ \Comment{Lemma~\ref{lem:q_sample} and Remark~\ref{rem:q_sample_whp}, fresh randomness for each $t$, $S_V^{(t)}\in \R^{n\times s_V^{(t)}}$}
\State $\wt V^{(t)}\gets S_V^{(t)\top}V$ \Comment{$\wt V^{(t)}\in \R^{s_V^{(t)}\times d}$}
\State $\wt H^{(t)}\gets \{S_E(i)S_V^{(t)}(j)\cdot \exp(\langle x_i,x_{j+n}\rangle) \}_{(i,j)\in S_E\times S_V^{(t)}}$ \Comment{$\wt H^{(t)}=S_E^\top E\begin{bmatrix}{\bf 0}_{n\times s_V^{(t)}}\\ S_V^{(t)}\end{bmatrix}$}
\State $\wt L^{(t)}\gets \wt M^\dagger (\wt H^{(t)}\wt V^{(t)})$ \Comment{$\wt L^{(t)}\in \R^{s_E\times d}$}
\EndFor
\EndProcedure
\State 
\Procedure{Query}{$i\in [n]$}
\State $b_i\gets \text{QRN}.\textsc{Query}(i)$ \Comment{Algorithm~\ref{alg:norm_est}}
\State $u_i\gets \{S_E(j)\cdot \exp(\langle x_i, x_j\rangle) \}_{j\in S_E}$ \Comment{$u_i\in \R^{s_E}$}
\State $r^{(t)}\gets \wt L^{(t)\top}u_i/b_i$ for all $t\in [T]$ \Comment{$r^{(t)}=(e_i^\top \wt D^{-1}\wt A S_V^{(t)}S_V^{(t)\top}V)^\top$}
\State $m^{(t)}\gets$ the $\lceil T/2\rceil$-th smallest of $\|r^{(t)}-r^{(t')}\|_2$ over $t'\in [T]$, for all $t\in [T]$
\State $t_i\gets \arg\min_{t\in [T]} m^{(t)}$ \Comment{Lemma~\ref{lem:amm_median}}
\State \Return $r^{(t_i)}$
\EndProcedure
\State {\bf end data structure}
\end{algorithmic}
\end{algorithm}

\begin{theorem}[Informal version of Theorem~\ref{thm:main_formal}]
\label{thm:main_informal}
Let $Q, K, V\in \R^{n\times d}$ be the query, key and value matrices, let $\epsilon\in (0,1)$ and $\lambda>0$. Let $E\in \R^{2n\times 2n}$ be the exponential kernel matrix on the dataset $Q\cup K$ and $s_\lambda$ be the statistical dimension of $E$ (Definition~\ref{def:stat_dim}). 
Assume that $\|D^{-1}\|<\frac{1}{\epsilon\|A\|+\lambda\sqrt n}$ and let $\beta=\frac{1}{1-(\epsilon\|A\|+\lambda\sqrt n)\|D^{-1}\|}$. Assume also that the matrix $U_2$ formed in Algorithm~\ref{alg:norm_est} has full column rank. There exists a quantum data structure that preprocesses $Q, K, V$ through only row queries to these matrices, maintains matrices $\wt D, \wt A$ and $T=O(\log n)$ sampling matrices $S_V^{(1)},\ldots,S_V^{(T)}$ of $V$ implicitly, and answers a query for row $i$ with a vector $\wt r_i=e_i^\top \wt D^{-1}\wt A S_V^{(t_i)}S_V^{(t_i)\top}V$ for an index $t_i\in [T]$ that it selects. With probability at least $1-1/\poly(n)$, the matrix $\wt R$ whose $i$-th row is $\wt r_i$ satisfies
\begin{align*}
    \|\wt R - \Att(Q, K, V)\|_F \leq & ~ 3\beta\cdot \|D^{-1}\|\cdot (\epsilon\|A\|_F+\lambda\sqrt n)\cdot \|V\|_F.
\end{align*}
Moreover, on the same event, the data structure has the specification
\begin{itemize}
    \item It preprocesses $Q, K, V$ in $\wt O(\epsilon^{-1}n^{0.5}s_\lambda)$ row queries to $Q, K$ and $\wt O(\epsilon^{-1}n^{0.5})$ row queries to $V$, and $\wt O(\epsilon^{-1}n^{0.5}(s_\lambda^{3}+s_\lambda^{2}d+d)+\epsilon^{-2}s_\lambda d)$ time.
    \item For any $i\in [n]$, it returns $\wt r_i$ in $\wt O(s_\lambda^2+s_\lambda d)$ time.
\end{itemize}
\end{theorem}

We pause to make some remarks on Theorem~\ref{thm:main_informal}. The preprocessing time scales with $n^{0.5}$, achieving a quadratic speedup with respect to $n$ over any classical algorithm. Several parameters merit further discussion, in particular the statistical dimension $s_\lambda$ and the approximation factor for $\|D^{-1}\|$, denoted by $\beta$. We summarize their relationships as functions of $\lambda$ in Table~\ref{tab:lambda}. 
The term $\lambda\sqrt n$ in the error bound is additive. It comes from the Nystr{\"o}m error $\|\wt A-A\|\leq \lambda$, which does not depend on $\epsilon$. Hence no choice of $\epsilon$ alone makes the error vanish. Choosing $\lambda\leq \epsilon\|A\|/\sqrt n$ gives an error of $O(\epsilon\beta\|D^{-1}\|\|A\|_F\|V\|_F)$, at the cost of a larger $s_\lambda$. The cost of sampling $V$ does not depend on the row distortion $\alpha\geq 1$ of $V$ (Definition~\ref{def:row_distortion}). Leverage score sampling of $V$ would incur this factor (Lemma~\ref{lem:amm}). We empirically verify that (1) the assumption on $\|D^{-1}\|$ is easy to satisfy with wide range of choices for $\epsilon$, (2) the Frobenius norm of $A$ is only a small constant factor of its spectral norm, (3) the upper bound $d/{\rm srank}(V)$ on $\alpha$ is a small constant and (4) the infinity norm of $A$ is only a small constant factor of its spectral norm, implying that in practice, the $\lambda \sqrt{n}$ term in the hypothesis on $\|D^{-1}\|$, which bounds the error of each normalization factor, is likely to be $O(\lambda)$. We refer to Appendix~\ref{sec:exp} for a more detailed section.

\begin{table}[!ht]
    \centering
    \begin{tabular}{|l||l|c|l|l|}
    \hline
        $\lambda$ & $s_\lambda$ & $\frac{1}{\epsilon\|A\| + \lambda \sqrt{n}}$ & $\beta$ & $\lambda\sqrt n$ \\ \hline
        $\uparrow$ & $\downarrow$ & $\downarrow$ & $\uparrow$ & $\uparrow$ \\ \hline
        $\downarrow$ & $\uparrow$ & $\uparrow$ & $\downarrow$ & $\downarrow$ \\ \hline
    \end{tabular}
    \caption{Parameters $s_\lambda$, $\frac{1}{\epsilon\|A\| + \lambda \sqrt{n}}$, $\beta$ and the additive error term $\lambda\sqrt n$ as functions of $\lambda$.}
    \label{tab:lambda}
\end{table}
\paragraph{Bit Complexity of Our Algorithm.} 

Our discussions and results above are grounded in the assumption that arithmetic operations are performed in infinite precision, while this is usually adopted in the analysis of classical algorithms, QRAM model only allows $O(\log n)$ qubits and $\poly(n)$ bits. In Section~\ref{sec:bit_complexity}, we provide a preliminary bit complexity analysis of our algorithm, in particular centering around the inversions of $S^\top ES+\lambda I$ in Algorithm~\ref{alg:q_nystrom}. We do not analyze the conditioning of the pseudoinverses or of $U_2^\top U_2$. To the best of our knowledge, there is no prior work on studying the bit complexity of numerical linear algebra operations in the QRAM model, and we leave a comprehensive analysis of bit complexity as a future direction.

\section{Related Work}
\paragraph{Transformers and Attention Mechanism.} Transformers~\citep{vsp+17} have been the driving force behind large language models~\citep{dclt19,bmr+20,tli+23,bce+23,tab+23,lfx+24}. They are sequence-to-sequence generative models, where the sequence length is typically denoted by $n$. The key architectural component that distinguishes transformers from earlier models is the attention mechanism, which computes a softmax over the pairwise interactions of query-key vectors. However, computing the full softmax distribution requires $\Omega(n^2)$ time, due to the size of the attention matrix. This quadratic dependency renders transformers inefficient for long sequences, motivating a rich body of work aimed at approximating attention in subquadratic time. These approaches can be broadly categorized into three main classes: (1) \emph{Pattern-based sparse attention}: only a subset of attention matrix entries are computed, with the subset determined by predefined patterns, such as sliding windows or graph-based sparsity structures~\citep{dkod20,kkl20,rsvg21,syy22,cgrs19,bpc20,aoa+20,zgd+20}. (2) \emph{Kernel-based linear attention}: these methods attempt to linearize the kernel by exploiting the identity ${\sf K}(x_i, x_j) = \langle \phi(x_i), \phi(x_j) \rangle$ for a feature map $\phi: \mathbb{R}^d \rightarrow \mathbb{R}^m$. When the kernel is exponential, exact computation requires $m = \infty$, so many heuristic approximations for $\phi$ have been proposed~\citep{kvpf20,cld+21,wlk+20,ppy+21} with $m = O(d)$. (3) \emph{Data structure-based attention}: these works design specialized data structures for approximating various components of attention. Examples include estimating the normalization factor via kernel density estimation (KDE)~\citep{zhdk23}, using hashing to identify large entries~\citep{hjk+24}, applying polynomial approximation methods under bounded input conditions~\citep{as23}, and other algorithmic innovations~\citep{kmz24,zhmk24,bsz24,syz24,kbkw25,cam+25,cls+25,iksw25}. Our work falls into the third category, as we design quantum data structures to approximate each of the matrices involved in the attention computation.

\paragraph{Quantum Machine Learning.} Given a machine learning problem, can we solve it faster on a quantum computer? The paradigm of using quantum mechanics to accelerate machine learning algorithms has sparked significant interest, leading to a wide array of results across diverse problem domains, including clustering~\citep{kllp19,xclj23}, classification~\citep{lcw19}, regression~\citep{cw23}, training neural networks~\citep{chl+19,klp20}, convex optimization~\citep{cclw20,vggw20_convex,lz22,sz23,zzf+24,wzl24}, mathematical programming~\citep{bkl+19,vggw20_sdp,vg19,kp20,kps21,vgl+21,ag23}, graph sparsification~\citep{aw22}, and recommender systems~\citep{kp17}. Among the key quantum techniques, Grover search~\citep{g96} plays a foundational role. It provides a quadratic speedup for database search problems: given a function $f: [n] \rightarrow \{0,1\}$, the goal is to list up to $m$ indices $i$ such that $f(i) = 1$. The Grover search algorithm requires oracle access to $f$ and can produce these $m$ indices using only $O(\sqrt{mn})$ queries, in contrast to the $O(n)$ queries required classically. Several variants of Grover search have been developed to suit different computational settings. In this paper, we use the probabilistic version: given a list of $n$ probabilities $p_1, \ldots, p_n \in [0,1]$, Grover search can be used to sample a list of indices where each $i$ is selected independently with probability $p_i$. By the standard analysis of Grover search, this sampling requires an expected number of $\widetilde{O}(\sqrt{n P})$ queries to the probability values $p_i$ where $P=\sum_{i=1}^n p_i$. Before our work,~\cite{gsyz23} also applied Grover search to accelerate attention computation. However, their method requires a structural assumption: for each query $q_i \in \mathbb{R}^d$, the associated set $S_i = \{ j \in [n] : \langle q_i, k_j \rangle \geq \tau \}$ must have cardinality at most $k$. Under this assumption, their algorithm runs in time $\widetilde{O}(n^{1.5} k^{0.5} d + nkd)$. Notably, if $k = n$, then their algorithm offers no speedup over the exact computation.

\section{Conclusion}\label{sec:conclusion}

We consider the problem of approximating the attention module in the row query model, where the goal is to return individual rows of the approximate attention matrix. We design a quantum data structure that preprocesses $Q$, $K$, and $V$ in $\widetilde{O}(\epsilon^{-1} n^{0.5} \, \mathrm{poly}(s_\lambda, d)+\epsilon^{-2}s_\lambda d)$ time, and answers any row query in $\widetilde{O}(s_\lambda^2 + s_\lambda d)$ time. To the best of our knowledge, this is the first quantum algorithm to achieve sublinear dependence on $n$ even in the row query model.

Our work also has several limitations, which raises interesting open questions. (1) The error guarantee we obtain is in Frobenius norm rather than spectral norm. While Frobenius norm bounds the sum of the squared $\ell_2$ errors across all rows, the spectral norm provides a worst-case guarantee that each row is well approximated. Therefore, it would be desirable to strengthen the result to achieve a spectral norm guarantee. (2) The error bound we obtain contains an additive $\lambda \sqrt{n}$ term, which does not shrink with $\epsilon$. It stems from turning the spectral-norm Nystr{\"o}m error $\|\wt A-A\|\leq \lambda$ into a bound of $\lambda\sqrt n$ on the $\ell_2$ error of the $n$ normalization factors (the row sums of $A$). This bound seems overly pessimistic, and it theoretically forces one to choose small value for $\lambda$, hindering the advantage of small statistical dimension. It would be interesting to remove the $\sqrt{n}$ factor in the additive term. (3) While we provide a preliminary bit complexity analysis of our algorithm in Section~\ref{sec:bit_complexity}, we feel a more comprehensive study of bit complexity of numerical linear algebra in the QRAM model is needed. We leave this as a major future direction, as it will significantly broaden the practicality of these quantum speedups. (4) Our algorithm in its current form can only compute the full attention \emph{without} the causal mask, as using the Nystr{\"o}m approximation implicitly assumes the complete interactions between queries and keys. To implement causal masking, one possibility is to design a quantum kernel density estimation data structure as shown in~\cite{zhdk23} classically.

\ifdefined\isarxiv

\paragraph{Roadmap.} In Section~\ref{sec:kernel}, we describe the quantum algorithm for exponential kernels. In Section~\ref{sec:norm}, we discuss how to estimate the normalization factor. In Section~\ref{sec:leverage}, we show the details on approximating matrix multiplication by row sampling. In Section~\ref{sec:result}, we combine things together and obtain the main result. In Section~\ref{sec:exp}, we empirically verify the assumptions on the parameters. In Section~\ref{sec:bit_complexity}, we discuss the bit complexity of our algorithm.

\section{Quantum Algorithm for Exponential Kernel}
\label{sec:kernel}
In this section, we give a generic reduction from attention matrix to a kernel matrix. Given queries and keys $Q=\{q_1,\ldots,q_n\},K=\{k_1,\ldots,k_n\}$, recall that we are interested in the matrix $\exp(QK^\top)$ where the $(i,j)$-th entry is $\exp(q_i^\top k_j)$, and this matrix is not a PSD kernel matrix. We show a reduction that first computes the exponential kernel ${\sf K}(x,y)=\exp(\langle x,y\rangle)$ over the dataset $Q\cup K$, then we can effectively extract certain blocks of the kernel matrix $E$ that approximates $\exp(QK^\top)$ well. We start with a lemma on block approximation.

\begin{lemma}\label{lem:block_approx}
Let $E\in \R^{2n\times 2n}$ be a PSD matrix and $E=\begin{bmatrix}
    B & A\\
    A^\top & C
\end{bmatrix}$ where each block is of size $n\times n$. Suppose there exists a matrix $\wt E\in \R^{2n\times 2n}$ such that $\wt E\preceq E\preceq \wt E+\lambda I$ for $\lambda>0$ and let $\wt E=\begin{bmatrix}
    \wt B & \wt A \\
    \wt A^\top & \wt C
\end{bmatrix}$, then we have
\begin{align*}
    \|A - \wt A\| \leq \lambda~\text{and}~\|A-\wt A\|_F\leq \lambda\sqrt{n}.
\end{align*}
\end{lemma}

\begin{proof}
We let $v\in \R^n$ be the vector that realizes the spectral norm $A-\wt A$, consider the augmented vector $\begin{bmatrix}
    {\bf 0}_n \\
    v
\end{bmatrix}$, then we see that
\begin{align*}
    \left\|(E-\wt E)\begin{bmatrix}
    {\bf 0}_n \\
    v
\end{bmatrix}\right\|_2^2 = & ~ \left\|\begin{bmatrix}
    (A-\wt A)v \\
    (C-\wt C)v
\end{bmatrix}\right\|_2^2 \\ 
= & ~ \|(A-\wt A)v\|_2^2+\|(C-\wt C)v\|_2^2\\
\leq & ~ \lambda^2,
\end{align*}
where the last step is by $\|E-\wt E\|\leq \lambda$ and our test vector is unit norm. As $\|(C-\wt C)v\|_2^2$ is trivially non-negative, we conclude that $\|(A-\wt A)v\|_2=\|A-\wt A\|\leq \lambda$, as desired. To obtain a Frobenius norm bound, note that $\|A-\wt A\|_F\leq \sqrt{n}\cdot \|A-\wt A\|\leq \lambda\sqrt{n}$.
\end{proof}
Our plan is to form the kernel matrix over the dataset $Q\cup K$ implicitly via Nystr{\"o}m approximation, then extract corresponding blocks to approximate $\exp(QK^\top)$. 

\begin{corollary}\label{cor:attention_approx}
Let $Q,K\in \R^{n\times d}$ and let $E\in \R^{2n\times 2n}$ be the exponential kernel matrix over the dataset $Q\cup K$, suppose there exists an $\wt E\in \R^{2n\times 2n}$ such that $\wt E\preceq E\preceq \wt E+\lambda I$ for some $\lambda>0$, then there exists $\wt A\in \R^{n\times n}$ such that
\begin{align*}
    \|\wt A-\exp(QK^\top)\|\leq  \lambda~\text{and}~\|\wt A-\exp(QK^\top)\|_F \leq \lambda\sqrt{n}.
\end{align*}
\end{corollary}

\begin{proof}
The result is a consequence of Lemma~\ref{lem:block_approx} by identifying that 
\begin{align*}
    E = & ~ \begin{bmatrix}
        \exp(QQ^\top) & \exp(QK^\top) \\
        \exp(KQ^\top) & \exp(KK^\top)
    \end{bmatrix},
\end{align*}
and $\wt E$ contains proper approximations for the desired blocks. 
\end{proof}

It remains to give an efficient algorithm to approximate the exponential kernel matrix $E$. A popular scheme is via Nystr{\"o}m approximation~\citep{ws00}: the algorithm selects a subset of ``landmark'' points, and constructs $\wt E$ through these landmarks.~\cite{mm17} uses recursive ridge leverage score sampling to generate such an approximation efficiently.~\cite{mm17} presents an algorithm that uses $\wt O(ns_\lambda\log(1/\delta))$ kernel function evaluations and $\wt O(ns_\lambda^2 \log(1/\delta))$ additional runtime to compute an approximation $\wt K$ satisfying $\wt K\preceq K\preceq \wt K+\lambda I$ with probability at least $1-\delta$. We restate their main result here for the sake of completeness.

\begin{lemma}[Theorem 7 of~\cite{mm17}]
\label{lem:classical_nystrom}
Let $s=O(s_\lambda \log(s_\lambda/\delta))$, there exists a weighted sampling matrix $S\in \R^{n\times s}$, such that the Nystr{\"o}m approximation of $E$, $\wt E=ES(S^\top ES)^\dagger S^\top E$ satisfies
\begin{align*}
    \wt E\preceq E\preceq \wt E+\lambda I,
\end{align*}
holds with probability at least $1-\delta$. Moreover, $S$ can be computed using $O(ns)$ kernel evaluations and $O(ns^2)$ additional time.
\end{lemma}

Our main contribution is a quantum algorithm that generates the approximation in \emph{sublinear time}. Before introducing the algorithm, we recall several key concepts.

Lemma~\ref{lem:classical_nystrom} relies on approximating the ridge leverage score on a sample, which can be captured by the notion of generalized ridge leverage score.

\begin{definition}[Generalized ridge leverage score,~\cite{mm17}]
\label{def:gen_ridge_ls}
Let $E\in \R^{n\times n}$ be a kernel matrix, let $\lambda>0$, and let $S\in \R^{n\times s}$ be any weighted sampling matrix, the $\lambda$-generalized ridge leverage score with respect to $S$, is defined for any $i\in [n]$,
\begin{align*}
    \wt \tau_i^\lambda := & ~ \frac{1}{\lambda} (E-ES(S^\top ES+\lambda I)^{-1}S^\top E)_{i,i},
\end{align*}
let $B\in \R^{n\times n}$ be any factorization of $E=BB^\top$, it can be equivalently defined as 
\begin{align*}
    \wt \tau_i^\lambda = & ~ b_i^\top (B^\top SS^\top B+\lambda I)^{-1}b_i,
\end{align*}
where $b_i$ is the $i$-th row of $B$.
\end{definition}

We also record a procedure introduced in~\cite{ag23} that generates a spectral approximation of an $n\times d$ matrix, given only queries to its rows, using quantum leverage score sampling. It is the first step of $\textsc{QMatVec}$ (Lemma~\ref{lem:q_matvec}). Leverage score sampling of $V$ would also need it (Remark~\ref{rem:no_alpha}).

\begin{lemma}[Theorem 3.1 of~\cite{ag23}]
\label{lem:q_spectral_approx}
Let $U\in \R^{n\times d}$, $\epsilon,\delta\in (0,1)$. There exists a quantum algorithm that computes a weighted sampling matrix $S\in \R^{n\times s}$ with $s=O(\epsilon^{-2}d\log(d/\delta))$ such that with probability at least $1-\delta$,
\begin{align*}
    (1-\epsilon)U^\top U \preceq U^\top SS^\top U \preceq (1+\epsilon)U^\top U.
\end{align*}
The quantum algorithm uses $\wt O(\epsilon^{-1}n^{0.5}d^{0.5})$ row queries to $U$, and it takes time $\wt O(\epsilon^{-1}n^{0.5}d^{1.5}+d^\omega)$.
\end{lemma}

We prove the key algorithmic result of this section.

\begin{theorem}
\label{thm:q_nystrom}
Let $\{x_1,\ldots,x_n\}\subseteq \R^d$ be a dataset, ${\sf K}:\R^d \times \R^d\rightarrow \R$ be a positive semidefinite kernel function, $\lambda>0$ and $\delta\in (0,1)$. Let $E$ be the kernel matrix where $E_{i,j}={\sf K}(x_i,x_j)$. Suppose $s=O(s_\lambda \log(n/\delta))$, then Algorithm~\ref{alg:q_nystrom} computes a weighted sampling matrix $S\in \R^{n\times s}$ such that with probability at least $1-\delta$,
\begin{align*}
    \wt E\preceq E \preceq \wt E+\lambda I,
\end{align*}
where $\wt E=ES(S^\top ES)^\dagger S^\top E$. Moreover, on the same event, the algorithm computes $S$ with $\wt O(n^{0.5}s^{0.5}\log(n/\delta))$ row queries to the dataset and in time $\wt O(n^{0.5}s^{1.5}\log(n/\delta)\cdot ({\cal T}_{\sf K}+s)+s^\omega)$, where ${\cal T}_{\sf K}$ is the time to evaluate the kernel function.
\end{theorem}

\begin{proof}
We note that the major differences between Algorithm~\ref{alg:q_nystrom} and the algorithm in~\cite{mm17} are
\begin{itemize}
    \item \cite{mm17} algorithm is recursive, our algorithm unrolls the recursion and iteratively constructs the weighted sampling matrix;
    \item \cite{mm17} computes all $p_i$'s classically, while we use $\textsc{QSample}$ to generate samples.
\end{itemize}
These two differences do not affect the analysis of~\cite{mm17}, because, run as in Remark~\ref{rem:q_sample_whp} with $\eta=\delta'$ for the $\delta'$ defined below, $\textsc{QSample}$ returns at each level its exact sample, which keeps each index independently with probability $p_i$, except with probability at most $\delta'$. We add these failure probabilities at the end of the proof. There is a third difference. \cite{mm17} sets
\begin{align*}
    p_i=\min\Big\{1,16\wt l_i\log\Big(\sum_j\wt l_j/\delta\Big)\Big\},
\end{align*}
where $\wt l_j$ are its approximate scores and the sum is over all indices of the current level. Computing this sum takes one oracle call per index of $S_t$, which is not sublinear. We use the fixed factor $L:=\log(8n^2/\delta)$ in place of $\log(\sum_j\wt l_j/\delta)$, and we now go through the analysis of~\cite{mm17} with this factor.

Let $\delta':=\delta/(8n)$, so that $L=\log(n/\delta')$ and $\delta'<1/8$, as the lemmas of~\cite{mm17} that we use require. We assume $s_\lambda\geq 1/2$, as~\cite{mm17} does in the proof of its Theorem 7. If $s_\lambda<1/2$, then $\lambda>\|E\|$, so the guarantee holds for every $S$, and the bounds below hold with $s_\lambda+1$ in place of $s_\lambda$.

For $t\in \{0,1,\ldots,T\}$, let $E^{(t)}$ be the kernel matrix of the points indexed by $S_t$, and let $\wt D_0:=I$. We prove by induction on $t$ that, with probability at least $1-(3t+1)\delta'$, the set $S_0$ has $O(\log(n/\delta))$ elements, each of $\wt D_1,\ldots,\wt D_t$ has $O(s_\lambda L)$ columns and
\begin{align}\label{eq:level_invariant}
    \frac{1}{2}(B_t^\top B_t+\lambda I)\preceq B_t^\top \wt D_t\wt D_t^\top B_t+\lambda I\preceq \frac{3}{2}(B_t^\top B_t+\lambda I)
\end{align}
for any matrix $B_t$ with $B_tB_t^\top=E^{(t)}$. This is the spectral guarantee in Theorem 8 of~\cite{mm17}, stated for the level $t$. In the case $t=0$, Eq.~\eqref{eq:level_invariant} is trivial. The size of $S_0$ is a sum of independent indicator variables with mean $n/2^T\leq m_0$, so a Chernoff bound gives $|S_0|\leq 2m_0+O(\log(1/\delta'))=O(\log(n/\delta))$ with probability at least $1-\delta'$. The case $t=T$ gives the first part of the theorem, because Eq.~\eqref{eq:level_invariant} for $t=T$ implies $\wt E\preceq E\preceq \wt E+\lambda I$ (see the proof of Theorem 7 of~\cite{mm17}), and $(3T+1)\delta'\leq \delta$ since $T\leq n$.

Fix $t\geq 1$ and assume the claim for $t-1$. Let $B$ be any matrix with $BB^\top=E^{(t)}$, let $b_i^\top$ be its rows, and let $\tau_i$ be the $\lambda$-ridge leverage score of $i\in S_t$ in $E^{(t)}$ (Definition~\ref{def:ridge_ls}). The matrix $E^{(t)}$ is a principal submatrix of $E$, so by Cauchy interlacing its eigenvalues are at most the corresponding eigenvalues of $E$, and $\sum_{i\in S_t}\tau_i=s_\lambda(E^{(t)})\leq s_\lambda$. Every $\tau_i$ is at most $1$, so we also have $\sum_{i\in S_t}\tau_i\leq n$.

\emph{The scores.} Let $\ov S$ be the unweighted sampling matrix of $S_{t-1}$ as a subset of $S_t$, and let $\wh S:=\ov S\wt D_{t-1}$, which is $D_{t-1}$ without its zero rows outside $S_t$. By Definition~\ref{def:gen_ridge_ls}, the oracle of level $t$ returns
\begin{align*}
    q_i=\frac{5}{\lambda}\cdot (E-ED_{t-1} (D_{t-1}^\top ED_{t-1}+\lambda I)^{-1} D_{t-1}^\top E)_{i,i}=5\cdot b_i^\top (B^\top \wh S\wh S^\top B+\lambda I)^{-1}b_i.
\end{align*}
Let $\ov\tau_i:=b_i^\top (B^\top \ov S\,\ov S^\top B+\lambda I)^{-1}b_i$. Since $\ov S\,\ov S^\top\preceq I$, we have $\ov\tau_i\geq \tau_i$. The matrix $\ov S^\top B$ satisfies $(\ov S^\top B)(\ov S^\top B)^\top=E^{(t-1)}$, so Eq.~\eqref{eq:level_invariant} for $t-1$ reads
\begin{align*}
    \frac{1}{2}(B^\top \ov S\,\ov S^\top B+\lambda I)\preceq B^\top \wh S\wh S^\top B+\lambda I\preceq \frac{3}{2}(B^\top \ov S\,\ov S^\top B+\lambda I).
\end{align*}
Inverting these matrices gives $\frac{2}{3}\ov\tau_i\leq q_i/5\leq 2\ov\tau_i$, and therefore
\begin{align}\label{eq:score_sandwich}
    \tau_i\leq q_i\leq 10\,\ov\tau_i\quad \text{for all } i\in S_t.
\end{align}

\emph{The spectral guarantee.} Lemma 11 of the arXiv version of~\cite{mm17} gives Eq.~\eqref{eq:level_invariant} with probability at least $1-\delta'$ when each index is sampled independently with the probability $p_i$ displayed above, with $\delta'$ in place of $\delta$ and with any scores $\wt l_i\geq \tau_i$, and the sampled columns are rescaled by $1/\sqrt{p_i}$. Its proof uses the logarithmic factor only through the bound
\begin{align}\label{eq:bernstein_condition}
    \frac{\wt l_i}{p_i}\leq \frac{1}{16\log(X/\delta')}\quad \text{for all } i \text{ with } p_i<1,
\end{align}
where $X\geq 1$ is an upper bound on the sum of the $\lambda$-ridge leverage scores. In~\cite{mm17}, $X$ is this sum itself. Our probabilities $p_i=\min\{1,16q_iL\}$ satisfy Eq.~\eqref{eq:bernstein_condition} with $\wt l_i=q_i$ and $X=n$, because $q_i\geq \tau_i$ by Eq.~\eqref{eq:score_sandwich}, $q_i/p_i=1/(16L)$ whenever $p_i<1$, and $L=\log(n/\delta')$. Hence $\wt D_t$ satisfies Eq.~\eqref{eq:level_invariant} with probability at least $1-\delta'$.

\emph{The sample size.} We bound $\sum_{i\in S_t}p_i$ by the reweighting argument in the proof of Lemma 5 of~\cite{mm17}, which we run with the threshold $\alpha:=1/(32L)$. That argument gives a diagonal matrix $W$ with $0\preceq W\preceq I$ such that every $\lambda$-ridge leverage score of $WE^{(t)}W$ is at most $\alpha$, and such that the set $R:=\{i\in S_t:W_{i,i}<1\}$ has at most $\frac{1}{\alpha}\sum_{i\in S_t}\tau_i\leq 32Ls_\lambda$ elements. The matrix $WB$ satisfies $(WB)(WB)^\top=WE^{(t)}W$, and the $\lambda$-ridge leverage scores of $WE^{(t)}W$ sum to at most $\sum_{i\in S_t}\tau_i$ because $B^\top W^2B\preceq B^\top B$. Each index of $S_t$ belongs to $S_{t-1}$ independently with probability $1/2$, and $\alpha/(1/2)=1/(16L)$. Hence the argument of the previous paragraph applies to $WB$ in place of $B$, with all $p_i$ equal to $1/2$, all $\wt l_i$ equal to $\alpha$ and $X=n$. With probability at least $1-\delta'$, it gives
\begin{align*}
    \frac{1}{2}(B^\top W^2B+\lambda I)\preceq 2B^\top W\ov S\,\ov S^\top WB+\lambda I.
\end{align*}
The matrices $W$ and $\ov S\,\ov S^\top$ are diagonal and $W\preceq I$, so $B^\top \ov S\,\ov S^\top B\succeq B^\top W\ov S\,\ov S^\top WB$, and on this event
\begin{align*}
    B^\top \ov S\,\ov S^\top B+\lambda I\succeq \frac{1}{2}(2B^\top W\ov S\,\ov S^\top WB+\lambda I)\succeq \frac{1}{4}(B^\top W^2B+\lambda I).
\end{align*}
Hence $\ov\tau_i\leq 4b_i^\top (B^\top W^2B+\lambda I)^{-1}b_i$ for all $i\in S_t$. For $i\notin R$, the right-hand side is $4$ times the $\lambda$-ridge leverage score of $i$ in $WE^{(t)}W$, so $\sum_{i\in S_t\setminus R}\ov\tau_i\leq 4\sum_{i\in S_t}\tau_i\leq 4s_\lambda$. We use $p_i\leq 1$ for $i\in R$ and Eq.~\eqref{eq:score_sandwich} for the other indices, and obtain
\begin{align*}
    \sum_{i\in S_t}p_i\leq |R|+16L\sum_{i\in S_t\setminus R}q_i\leq 32Ls_\lambda+16L\cdot 10\cdot 4s_\lambda=O(s_\lambda L).
\end{align*}
The number of columns of $\wt D_t$ is a sum of independent indicator variables with mean $\sum_{i\in S_t}p_i$. By a Chernoff bound, it is at most $2\sum_{i\in S_t}p_i+O(\log(1/\delta'))=O(s_\lambda L)$ with probability at least $1-\delta'$, where we used $s_\lambda\geq 1/2$.

The three events of level $t$ fail with probability at most $3\delta'$ in total, which completes the induction. Since $L$ and $|S_0|$ are both $O(\log(n/\delta))$, all sample sizes are $O(s_\lambda\log(n/\delta))=O(s)$. It remains to give a bound on the running time. 
\begin{itemize}
    \item Computing $M_0$: $M_0\in \R^{|S_0|\times |S_0|}$ contains the values of kernel functions over $|S_0|^2=O(\log^2(n/\delta))$ pairs, forming it takes $\wt O(1)\cdot {\cal T}_{\sf K}$ time;
    \item Computing $\wh M$: we maintain the invariant that $M_t$ has $O(s)$ rows and columns for all $t\in [T]$, therefore computing $\wh M$ is inverting an $O(s)\times O(s)$ matrix, which takes $O(s^\omega)$ time;
    \item Computing $D_{t-1}^\top K_i$: this operation involves computing $s$ weighted kernel function evaluations, given $D_{t-1}$ stores a list of $s$ indices together with weights, it can be done in $O(s)\cdot {\cal T}_{\sf K}$ time;
    \item Oracle for $q_i$: for any fixed $i$, note that we need to form $D_{t-1}^\top K_i$ using $O(s)\cdot {\cal T}_{\sf K}$ time, and computing the quadratic form takes $O(s^2)$ time. Thus each oracle call takes $O(s)\cdot {\cal T}_{\sf K}+O(s^2)$ time;
    \item Computing $\wt D_t$: this step requires to compute at most $n$ probabilities, and each probability $p_i=\min\{1,16q_iL\}$ can be computed via one oracle call in $O(s)\cdot {\cal T}_{\sf K}+O(s^2)$ time, because the factor $L$ is fixed. We have shown above that $\sum_{i\in S_t}p_i=O(s_\lambda L)=O(s)$. We run $\textsc{QSample}$ as in Remark~\ref{rem:q_sample_whp} with $\eta=\delta'$, so that it repeats every search $r=O(\log(n/\delta'))=O(L)$ times. Since $\log(n/\delta')=L=O(s_\lambda L)$ by $s_\lambda\geq 1/2$, with probability at least $1-\delta'$ it makes $\wt O(n^{0.5}s^{0.5}L)$ oracle calls and returns its exact sample, thus the runtime is $\wt O(n^{0.5}s^{1.5}L\cdot ({\cal T}_{\sf K}+s))$.
\end{itemize}
Finally, every iteration considers at most $n$ points, there are $T=O(\log n)$ iterations and $L=O(\log(n/\delta))$, so we conclude the overall runtime of Algorithm~\ref{alg:q_nystrom} is
\begin{align*}
    \wt O(n^{0.5}s^{1.5}\log(n/\delta)\cdot ({\cal T}_{\sf K}+s)+s^\omega),
\end{align*}
as desired. On the events of the induction, $\sum_{i\in S_t}p_i=O(s_\lambda L)$ at every level, so each of the $T$ runs of $\textsc{QSample}$ exceeds its bound or does not return its exact sample with probability at most $\delta'$. If no run fails in this way, the algorithm agrees with the one that uses the exact samples, to which the induction applies. Hence the guarantee and the running time hold together with probability at least $1-(4T+1)\delta'\geq 1-\delta$, since $T\leq n$.
\end{proof}

We can then apply Theorem~\ref{thm:q_nystrom} to exponential kernel function and the dataset $Q\cup K$ to compute a Nystr{\"o}m sampling matrix $S$.

\begin{corollary}
\label{cor:approx_exp}
Let $Q,K\in \R^{n\times d}$, $\lambda>0$ and $\delta\in (0,1)$. Define the dataset $X=\{x_1,x_2,\ldots,x_{2n}\}\subseteq \R^d$ where for $i\in [n]$, $x_i=q_i$ and for $i\in \{n+1,\ldots,2n\}$, $x_i=k_{i-n}$. Let $E$ be the kernel matrix where $E_{i,j}=\exp(\langle x_i,x_j\rangle)$. Suppose $s=O(s_\lambda \log(n/\delta))$, then there exists an algorithm that computes a weighted sampling matrix $S\in \R^{2n\times s}$ such that, let $\wt E=ES(S^\top ES)^\dagger S^\top E$, then with probability at least $1-\delta$, $\wt E\preceq E\preceq \wt E+\lambda I$. Moreover, on the same event, $S$ is computed in $\wt O(n^{0.5}s^{1.5}\log(n/\delta)\cdot (d+s)+s^\omega)$ time.
\end{corollary}

\begin{proof}
Apply Theorem~\ref{thm:q_nystrom} to the kernel function ${\sf K}(x_i,x_j)=\exp(\langle x_i,x_j\rangle)$ and note that the kernel function can be computed in $O(d)$ time.
\end{proof}
\section{Estimating the Normalization Factor}
\label{sec:norm}
Given a sublinear quantum algorithm to approximate the matrix $\exp(QK^\top)$, our next step is to estimate the normalization factor $\exp(QK^\top){\bf 1}_n$ to compute the softmax matrix. We first show that given a Nystr{\"o}m approximation to the $2n\times 2n$ kernel matrix $E$, how to compute the normalization factor and the approximate guarantees.

\begin{lemma}\label{lem:infty_spectral_relation}
Let $M\in \R^{n\times n}$, then we have
\begin{align*}
    \|M\|_\infty \leq & ~ \sqrt n\cdot \|M\|.
\end{align*}
\end{lemma}

\begin{proof}
Fix any $i\in [n]$, we examine the row $M_{i,*}$, set the test vector $x$ to be $x_j=\begin{cases}
    +1, & \text{if $M_{i,j}\geq 0$}, \\
    -1 , & \text{otherwise}.
\end{cases}$, then 
\begin{align*}
    \|M_{i,*}\|_1 = & ~ M_{i,*}^\top x \\
    = & ~ \langle M^\top e_i, x\rangle \\
    \leq & ~ \|M^\top e_i\|_2 \cdot \|x\|_2 \\
    \leq & ~ \|M\|\cdot \|x\|_2 \\
    = & ~ \sqrt n\cdot \|M\|.
\end{align*}
The conclusion can be achieved by noting that this bound works for any row $i$.
\end{proof}

There are two major issues for estimating the normalization factor:
\begin{itemize}
    \item Corollary~\ref{cor:approx_exp} only allows us to compute the sampling matrix in sublinear time, explicitly forming the Nystr{\"o}m approximation $\wt E$ however, would require $\Omega(n)$ time since the matrix is of size $n\times n$;
    \item Even though we are given the explicit factorization $\wt E=UU^\top$ where $U\in \R^{2n\times s}$, we would have to compute $n$ normalization factors, which would require $\Omega(n)$ time.
\end{itemize}
In other words, because the output has size $\Omega(n)$, one cannot expect any quantum algorithm to run in $o(n)$ time. Instead, we design a quantum data structure with preprocessing time $o(n)$ time, and can support query to compute the normalization factor to any row efficiently. 

In particular, we are interested in the following algorithmic task: given query access to the rows of a matrix $U\in \R^{n\times s}$ and a vector $v\in \R^n$, output a vector $\wt \mu\in \R^s$ such that $\|\wt \mu-U^\top v\|_{(U^\top U)^{-1}}\leq \epsilon$, which can be solved via Lemma~\ref{lem:q_matvec}. For our application, $\|v\|_\infty=1$. However, we are interested in the quantity $UU^\top v$ so we need to measure the error $\|U(\wt \mu-U^\top v)\|_2$. How would a bound on the $\|\cdot \|_{(U^\top U)^{-1}}$ be useful? We prove a structural lemma below.

\begin{lemma}\label{lem:U1_U2_approx}
Let $x\in \R^s$ and let $U_2\in \R^{n\times s}$ have full column rank, satisfying $\|x\|_{(U_2^\top U_2)^{-1}}\leq \epsilon$ for some $\epsilon\in (0,1)$, let $U_1\in \R^{n\times s}$, we have 
\begin{align*}
    \|U_1x\|_2 \leq & ~ \epsilon\cdot \|U_1U_2^\top\|.
\end{align*}
\end{lemma}

\begin{proof}
We define the vector $y=(U_2^\top U_2)^{-1}x$ and $z=U_2y$, then
\begin{align*}
    \|z\|_2^2 = & ~ y^\top U_2^\top U_2 y \\
    = & ~ x^\top (U_2^\top U_2)^{-1} x \\
    = & ~ \|x\|_{(U_2^\top U_2)^{-1}}^2 \\
    \leq & ~ \epsilon^2,
\end{align*}
moreover, the vector of interest is $U_1 x$ which is
\begin{align*}
    U_1 x = & ~ U_1 (U_2^\top U_2) y \\
    = & ~ (U_1U_2^\top) U_2 y \\
    = & ~ (U_1U_2^\top) z,
\end{align*}
subsequently its $\ell_2$ norm can be bounded as
\begin{align*}
    \|U_1 x\|_2 \leq & ~ \|U_1U_2^\top\|\cdot \|z\|_2 \\
    \leq & ~ \epsilon\cdot \|U_1U_2^\top\|,
\end{align*}
as desired.
\end{proof}

\begin{corollary}\label{cor:mat_vec_norm}
Let $\epsilon\in (0,1), U_1, U_2\in \R^{n\times s}$ where $U_2$ has full column rank and $\wt A=U_1U_2^\top$, $v\in \R^n$, suppose there exists a vector $\wt \mu\in \R^s$ with $\|\wt \mu-U_2^\top v\|_{(U_2^\top U_2)^{-1}}\leq \epsilon$, then we have
\begin{align*}
    \|\wt Av-U_1 \wt \mu\|_2 \leq & ~ \epsilon\cdot \|U_1U_2^\top\|.
\end{align*}
\end{corollary}

\begin{proof}
We will apply Lemma~\ref{lem:U1_U2_approx} by setting $x=\wt \mu-U_2^\top v$ and by noting that $U_1x=U_1\wt \mu-U_1U_2^\top v=U_1\wt \mu-\wt Av$. 
\end{proof}

We are now in the position to state our formal theorem, which provides an end-to-end guarantee on estimating the normalization factor. For simplicity, we will prove the statement with high probability guarantee, i.e., the success probability is $1-1/\poly(n)$.

\begin{theorem}\label{thm:q_norm}
Let $Q, K\in \R^{n\times d}$, $\lambda>0$ and $\epsilon\in (0,1)$. Let $s=\wt O(s_\lambda)$ where $s_\lambda$ is the statistical dimension of the exponential kernel on $Q\cup K$. Assume that the matrix $U_2$ formed in Algorithm~\ref{alg:norm_est} has full column rank. This holds, for example, when the keys and the sampled points are in general position and duplicate samples are merged. Merging duplicate samples does not change $\wt E$. There exists a data structure (Algorithm~\ref{alg:norm_est}) with the following specification:
\begin{itemize}
    \item With probability at least $1-1/\poly(n)$, preprocessing in $\wt O(n^{0.5}s/\epsilon)$ row queries to $Q, K$ and time $\wt O(n^{0.5}s^{2}(s+d)/\epsilon+s^\omega)$.
    \item For any $i\in [n]$, it outputs an approximate normalization factor for row $i$ in time $O(s(s+d))$.
\end{itemize}
Moreover, with probability at least $1-1/\poly(n)$, it holds that for any $i\in [n]$, the output $b_i$ satisfies
\begin{align*}
    |b_i-\exp(q_iK^\top){\bf 1}_n | \leq & ~ O(\epsilon \|A\|+\lambda \sqrt n),
\end{align*}
if $\frac{\lambda\sqrt n}{\|A\|}\leq 1$, then the bound can be further simplified to
\begin{align*}
    |b_i-\exp(q_iK^\top){\bf 1}_n | \leq & ~ O(\lambda\sqrt n),
\end{align*}
and the preprocessing uses $\wt O(s\|A\|/\lambda)$ row queries and time
\begin{align*}
    \wt O(s^{2}(s+d) \|A\|/\lambda+s^\omega).
\end{align*}
\end{theorem}

\begin{proof}
Given $Q, K$, let $E\in \R^{2n\times 2n}$ be the associated exponential kernel matrix. We will first invoke Corollary~\ref{cor:approx_exp} to compute a sampling matrix $S\in \R^{2n\times s}$ where $s=\wt O(s_\lambda)$ such that $\wt E=ES(S^\top ES)^\dagger S^\top E$ approximates $E$, with probability at least $1-1/\poly(n)$ in time $\wt O(n^{0.5}s^{1.5}(s+d)+s^\omega)$, since its failure probability is $1/\poly(n)$ and hence $\log(n/\delta)=O(\log n)$. Set $U=ES(S^\top ES)^{\dagger/2}$, we have that $\wt E=UU^\top$ for $U=\begin{bmatrix}
    U_1 \\
    U_2
\end{bmatrix}$ with $U_1, U_2\in \R^{n\times s}$, and our desired approximate block for $A$ is $\wt A=U_1U_2^\top$. Note that forming $U$ explicitly would take $\Omega(n)$ time, so we instead implement a row oracle for $U_2$. Since $U_2\in \R^{n\times s}$, we only need to compute $s$ entries for each row, and let $N=(S^\top ES)^{\dagger/2}$, we see that $(U_2)_{j,*}=N (ES)_{j+n,*}$ and $(ES)_{j+n,*}$ contains values in the form of $S_k\cdot \exp(\langle x_{j+n},x_k\rangle)$ for $k\in S$. $N$ can be computed in $O(s^2 d+s^\omega)$ time, and row oracle for any $j\in [n]$ can be implemented in $O(sd+s^2)$ time. By Lemma~\ref{lem:q_matvec} with $U_2\in \R^{n\times s}$ in place of $A$ and $\|v\|_\infty=1$, $\wt \mu$ can be computed with $\wt O(n^{0.5}s/\epsilon)$ row queries to $U_2$ and $\wt O(n^{0.5}s^3/\epsilon+s^\omega)$ time, counting each row query as $O(s)$ time. Each row query to $U_2$ is answered with one row query to $K$ and $O(s(s+d))$ time, so $\wt \mu$ costs $\wt O(n^{0.5}s/\epsilon)$ row queries to $K$ and $\wt O(n^{0.5}s^2(s+d)/\epsilon+s^\omega)$ time.  To query the normalization factor for row $i$, note that it can be computed via $(U_1\wt \mu)_i=\langle (U_1)_{i,*}, \wt \mu\rangle$, which can be computed using row oracle, in $O(s(s+d))$ time. Thus, with probability at least $1-1/\poly(n)$, the overall runtime of our procedure can be summarized as
\begin{itemize}
    \item Preprocessing time $\wt O(n^{0.5}s^{2}(s+d)/\epsilon+s^\omega)$, with $\wt O(n^{0.5}s/\epsilon)$ row queries to $Q, K$.
    \item Query time $O(s(s+d))$.
\end{itemize}
It remains to give an approximation guarantee. With probability at least $1-1/\poly(n)$, we have $\|A-\wt A\|\leq \lambda$, and observe that
\begin{align*}
    |\wt a_i^\top {\bf 1_n}-\exp(q_iK^\top){\bf 1}_n| \leq & ~ \|(\wt A-A)v\|_\infty \\
    \leq & ~ \|\wt A-A\|_\infty\cdot \|v\|_\infty \\
    \leq & ~ \lambda \sqrt n,
\end{align*}
where $v={\bf 1}_n$, the second step is by the matrix infinity norm is the induced norm of vector $\ell_\infty$ norm, and the last step is by Lemma~\ref{lem:infty_spectral_relation} applied to the $n\times n$ matrix $\wt A-A$, together with $\|\wt A-A\|\leq \lambda$. On the other hand, our final output $b_i$ is an approximation to $\wt a_i^\top {\bf 1}_n$. Let $\wt y:=U_1\wt \mu$, by Corollary~\ref{cor:mat_vec_norm}, we have
\begin{align*}
    \|\wt A v-\wt y\|_2 \leq & ~ \epsilon\cdot \|\wt A\|,
\end{align*}

this holds with probability at least $1-\delta$, conditioning on this event, note that by Lemma~\ref{lem:block_approx}, we have that $\|\wt A\|\leq \|A\|+\lambda$. Thus, we conclude our final result by
\begin{align*}
    |b_i-\exp(q_iK^\top){\bf 1}_n| \leq & ~ |b_i-\wt a_i{\bf 1}_n|+|\wt a_i^\top {\bf 1}_n-\exp(q_iK^\top){\bf 1}_n| \\
    \leq & ~ \|\wt Av-\wt y\|_2+\lambda\sqrt n \\
    \leq & ~ \epsilon\cdot (\lambda+\|A\|)+\lambda \sqrt n.
\end{align*}

Now, suppose $\lambda\sqrt n\leq \|A\|$, then we could set $\epsilon=\frac{\lambda \sqrt n}{\|A\|}$, the error bound simplifies to $O(\lambda\sqrt n)$, and since $n^{0.5}/\epsilon=\|A\|/\lambda$, the preprocessing uses $\wt O(s\|A\|/\lambda)$ row queries and $\wt O(s^2(s+d)\|A\|/\lambda+s^\omega)$ time.
\end{proof}
\section{Approximate Matrix Products by Row Sampling}
\label{sec:leverage}
It remains to handle the value matrix, and we will do so via a machinery called approximate matrix multiplication. 

\begin{definition}[Approximate matrix multiplication,~\cite{cw17}]
Let $A\in \R^{n\times d}, B\in \R^{n\times m}$ and let $C=A^\top B\in \R^{d\times m}$. The approximate matrix multiplication problem asks to design a random matrix $S\in \R^{n\times s}$, such that
\begin{align*}
    \Pr[\|A^\top SS^\top B-C\|_F \leq \epsilon \|A\|_F\|B\|_F] \geq & ~ 1-\delta,
\end{align*}
where $\epsilon,\delta\in (0,1)$. We call such $S$ satisfying $(\epsilon,\delta)$-AMM.
\end{definition}

We generate the random matrix $S$ by sampling rows of $V$. Our algorithm samples them by their squared norms~\citep{dkm06}, and for comparison we also treat leverage score sampling. The standard proof for leverage score sampling (see, e.g.,~\cite{cw17}) requires $V$ to have orthonormal columns. For a general $V$, the sample size of leverage score sampling grows by a factor that we call the \emph{row distortion}.

\begin{definition}[Row distortion]
\label{def:row_distortion}
Let $A\in \R^{n\times d}$ for $n\geq d$, we define the row distortion of $A$, denoted by $\alpha(A)$, as
\begin{align*}
    \alpha(A) := & ~ \frac{d}{\|A\|_F^2}\cdot \max_{i\in [n]} \frac{\|a_i\|_2^2}{\tau_i},
\end{align*}
where $a_i$ is the $i$-th row of $A$ and $\tau_i$ is the $i$-th leverage score of $A$ (Definition~\ref{def:ls}), and the maximum is over the rows with $\tau_i>0$ (the other rows are zero). When $A$ is clear from context, we use $\alpha$ as an abbreviation.
\end{definition}

\begin{lemma}\label{lem:row_distortion_bounds}
Let $A\in \R^{n\times d}$ with $n\geq d$, then the row distortion of $A$ satisfies
\begin{align*}
    1\leq \alpha(A) \leq & ~ \frac{d}{{\rm srank}(A)},
\end{align*}
where ${\rm srank}(A)=\frac{\|A\|_F^2}{\|A\|^2}$ is the stable rank of $A$.
\end{lemma}

\begin{proof}
We derive an upper bound on $\|a_i\|_2^2$, let $A=U\Sigma V^\top$ be its SVD, then
\begin{align*}
    \|a_i\|_2^2 = & ~ \|e_i^\top U\Sigma V^\top \|_2^2 \\
    \leq & ~ \|u_i\|_2^2\cdot \|\Sigma V^\top\|^2 \\
    = & ~ \tau_i\cdot \|U\Sigma V^\top\|^2 \\
    = & ~ \tau_i\cdot \|A\|^2, 
\end{align*}
where the third step is by the definition of leverage score and spectral norm is unitary invariant. We thus obtain the following bound on $\alpha(A)$:
\begin{align*}
    \alpha(A) = & ~ \frac{d}{\|A\|_F^2}\cdot \max_{i\in [n]} \frac{\|a_i\|_2^2}{\tau_i} \\
    \leq & ~ \frac{d}{\|A\|_F^2}\cdot \max_{i\in [n]} \frac{\tau_i\cdot \|A\|^2}{\tau_i} \\
    = & ~ d\cdot \frac{\|A\|^2}{\|A\|_F^2} \\
    = & ~ \frac{d}{{\rm srank}(A)},
\end{align*}
where we recall that ${\rm srank}(A)=\frac{\|A\|_F^2}{\|A\|^2}$. For the lower bound, recall that $\sum_{i=1}^n\tau_i={\rm rank}(A)\leq d$, and that $\tau_i=0$ only for zero rows. Hence
\begin{align*}
    \|A\|_F^2 = \sum_{i:\,\tau_i>0}\frac{\|a_i\|_2^2}{\tau_i}\cdot \tau_i \leq \max_{i:\,\tau_i>0}\frac{\|a_i\|_2^2}{\tau_i}\cdot \sum_{i=1}^n\tau_i \leq d\cdot \max_{i:\,\tau_i>0}\frac{\|a_i\|_2^2}{\tau_i},
\end{align*}
which is $\alpha(A)\geq 1$.
\end{proof}

We now bound the error of independent row sampling. The bound covers both squared row norm sampling and leverage score sampling, the latter for matrices without orthonormal columns.

\begin{lemma}\label{lem:amm}
Let $A\in \R^{n\times d}$ be nonzero with rows $a_1,\ldots,a_n$, let $B\in \R^{n\times m}$ have rows $b_1,\ldots,b_n$, and let $s,\gamma>0$. Let $S$ be a weighted sampling matrix that keeps each row $i\in [n]$ independently with probability $p_i$ (so the number of columns of $S$ is random), where
\begin{align*}
    p_i\geq \min\Big\{1,\frac{s\|a_i\|_2^2}{\gamma\|A\|_F^2}\Big\}.
\end{align*}
Then
\begin{align*}
    \E\|A^\top SS^\top B-A^\top B\|_F^2\leq \frac{\gamma}{s}\|A\|_F^2\|B\|_F^2 .
\end{align*}
Consequently, for $\epsilon,\delta\in (0,1)$, $S$ is an $(\epsilon,\delta)$-AMM whenever $s\geq \gamma/(\epsilon^2\delta)$. In particular, $s=O(\epsilon^{-2}\gamma)$ gives an $(\epsilon,1/4)$-AMM. The hypothesis holds with $\gamma=\alpha(A)$ for leverage score sampling, $p_i\geq \min\{1,s\tau_i/d\}$, and with $\gamma=c$ for squared row norm sampling, $p_i=\min\{1,s\|a_i\|_2^2/Z\}$ with $0<Z\leq c\|A\|_F^2$.
\end{lemma}

\begin{proof}
Let $\xi_i\in \{0,1\}$ be the independent indicators that row $i$ is kept, so that
\begin{align*}
    SS^\top={\rm diag}(\xi_1/p_1,\ldots,\xi_n/p_n),
\end{align*}
where we set $\xi_i/p_i:=0$ if $p_i=0$. The hypothesis gives $p_i>0$ for every nonzero row $a_i$, and the rows with $a_i=0$ do not contribute below. Thus
\begin{align*}
    X:=A^\top SS^\top B-A^\top B=\sum_{i:\,a_i\neq 0}\Big(\frac{\xi_i}{p_i}-1\Big)a_ib_i^\top .
\end{align*}
The coefficients $\xi_i/p_i-1$ are independent, have mean $0$, and have variance $1/p_i-1$. In $\E\|X\|_F^2=\E\tr[X^\top X]$ the cross terms therefore vanish, and $\|a_ib_i^\top\|_F=\|a_i\|_2\|b_i\|_2$, so
\begin{align*}
    \E\|X\|_F^2=\sum_{i:\,a_i\neq 0}\Big(\frac{1}{p_i}-1\Big)\|a_i\|_2^2\|b_i\|_2^2 .
\end{align*}
A row with $p_i=1$ contributes $0$. If $p_i<1$, then the minimum in the hypothesis is less than $1$, so $p_i\geq s\|a_i\|_2^2/(\gamma\|A\|_F^2)$ and $(1/p_i-1)\|a_i\|_2^2\leq \gamma\|A\|_F^2/s$. Summing over $i$ gives
\begin{align*}
    \E\|X\|_F^2\leq \frac{\gamma}{s}\|A\|_F^2\sum_{i=1}^n\|b_i\|_2^2=\frac{\gamma}{s}\|A\|_F^2\|B\|_F^2 .
\end{align*}
By Markov's inequality applied to $\|X\|_F^2$,
\begin{align*}
    \Pr\big[\|X\|_F>\epsilon\|A\|_F\|B\|_F\big]\leq \frac{\gamma}{s\epsilon^2}\leq \delta
\end{align*}
whenever $s\geq \gamma/(\epsilon^2\delta)$. For leverage score sampling, the definition of $\alpha=\alpha(A)$ (Definition~\ref{def:row_distortion}) gives $\tau_i/d\geq \|a_i\|_2^2/(\alpha\|A\|_F^2)$ for every row with $\tau_i>0$, and the other rows are zero, so
\begin{align*}
    \min\{1,s\tau_i/d\}\geq \min\{1,s\|a_i\|_2^2/(\alpha\|A\|_F^2)\}.
\end{align*}
For squared row norm sampling, $Z\leq c\|A\|_F^2$ gives $s\|a_i\|_2^2/Z\geq s\|a_i\|_2^2/(c\|A\|_F^2)$.
\end{proof}

\begin{remark}\label{rem:amm_tight}
For these sampling distributions, the dependence on $\delta$ in Lemma~\ref{lem:amm} cannot be improved to $\log(1/\delta)$. Let $n\geq 2$ and $d=m=1$, so that $A=a\in \R^n$ and both distributions are $p_i=\min\{1,sa_i^2/\|a\|_2^2\}$, which is the case $\gamma=1$. Let $B=e_1$. If $s\epsilon^2\geq 1$, choose $a$ with $a_1=\|a\|_2/(2s\epsilon)$. Then $p_1=1/(4s\epsilon^2)<1$. If row $1$ is kept, the estimate of $A^\top B=a_1$ is $a_1/p_1=2\epsilon\|a\|_2$, so the error is at least $2\epsilon\|a\|_2-a_1\geq \frac{3}{2}\epsilon\|a\|_2>\epsilon\|A\|_F\|B\|_F$. Hence the failure probability is at least $p_1=1/(4s\epsilon^2)$. If instead $1\leq s<\epsilon^{-2}$, choose $a$ with $a_1=\|a\|_2/(2\sqrt s)$. Then $p_1=1/4$, and if row $1$ is kept, the error is $a_1/p_1-a_1=3a_1=\frac{3}{2}\|a\|_2/\sqrt s>\epsilon\|a\|_2$, so the failure probability is at least $1/4$. Hence failure probability $\delta<1/4$ requires $s\geq 1/(4\epsilon^2\delta)$. This is why Algorithm~\ref{alg:q_attention} draws $O(\log n)$ independent samples and selects among them (Lemma~\ref{lem:amm_median}) instead of drawing one larger sample.
\end{remark}

We now recover a high-probability guarantee from independent samples. A sample cannot be checked without computing the product it approximates, so the selection is made row by row, at query time.

\begin{lemma}\label{lem:amm_median}
Let $V\in \R^{n\times d}$, let $W\in \R^{n\times n}$ have rows $w_1^\top,\ldots,w_n^\top$, and let $\epsilon,\delta\in (0,1)$. Let $T\geq C\ln(n/\delta)$ be an integer, where $C$ is a large enough absolute constant, and let $S^{(1)},\ldots,S^{(T)}$ be independent random weighted sampling matrices, independent of $W$, such that for every fixed $w\in \R^n$ and every $t\in [T]$,
\begin{align*}
    \Pr\Big[\|w^\top S^{(t)}S^{(t)\top}V-w^\top V\|_2>\frac{\epsilon}{3}\|w\|_2\|V\|_F\Big]\leq \frac{1}{4}.
\end{align*}
For $i\in [n]$ and $t\in [T]$, let $r_i^{(t)}:=w_i^\top S^{(t)}S^{(t)\top}V$, let $m_i^{(t)}$ be the $\lceil T/2\rceil$-th smallest of the $T$ numbers $\|r_i^{(t)}-r_i^{(t')}\|_2$ with $t'\in [T]$, and let $t_i\in [T]$ minimize $m_i^{(t)}$ over $t\in [T]$, with ties broken arbitrarily. Let $\wt R\in \R^{n\times d}$ be the matrix whose $i$-th row is $r_i^{(t_i)}$. Then, with probability at least $1-\delta$,
\begin{align*}
    \|\wt R-WV\|_F\leq \epsilon\|W\|_F\|V\|_F .
\end{align*}
\end{lemma}

\begin{proof}
We condition on $W$, which does not change the distribution of $S^{(1)},\ldots,S^{(T)}$. Fix $i\in [n]$, let $y_i:=w_i^\top V$ and $\rho_i:=\frac{\epsilon}{3}\|w_i\|_2\|V\|_F$, and call $t\in [T]$ good for row $i$ if $\|r_i^{(t)}-y_i\|_2\leq \rho_i$. By the hypothesis with $w=w_i$, each $t$ is bad for row $i$ independently with probability at most $1/4$. By Hoeffding's inequality, at least $T/2$ indices are bad for row $i$ with probability at most $\exp(-2T(1/4)^2)=\exp(-T/8)$, which is at most $\delta/n$ for $C\geq 8$. By a union bound over $i\in [n]$, with probability at least $1-\delta$, the set $G_i$ of indices that are good for row $i$ has $|G_i|>T/2$ for every $i$. We assume this event and fix $i$.

For all $t,t'\in G_i$, the triangle inequality gives $\|r_i^{(t)}-r_i^{(t')}\|_2\leq 2\rho_i$. Since $|G_i|\geq \lceil T/2\rceil$, this shows $m_i^{(t)}\leq 2\rho_i$ for every $t\in G_i$, and hence $m_i^{(t_i)}\leq 2\rho_i$. By the definition of $m_i^{(t_i)}$, the set $N_i:=\{t'\in [T]:\|r_i^{(t_i)}-r_i^{(t')}\|_2\leq m_i^{(t_i)}\}$ has $|N_i|\geq \lceil T/2\rceil$. Since $|G_i|+|N_i|>T$, some $t'$ lies in $G_i\cap N_i$, and
\begin{align*}
    \|r_i^{(t_i)}-y_i\|_2\leq \|r_i^{(t_i)}-r_i^{(t')}\|_2+\|r_i^{(t')}-y_i\|_2\leq 2\rho_i+\rho_i=\epsilon\|w_i\|_2\|V\|_F .
\end{align*}
Squaring and summing over $i\in [n]$ gives $\|\wt R-WV\|_F^2\leq \epsilon^2\|V\|_F^2\sum_{i=1}^n\|w_i\|_2^2=\epsilon^2\|W\|_F^2\|V\|_F^2$.
\end{proof}

It remains to sample by squared row norms in sublinear time, which needs an estimate of $\|V\|_F^2$ within a constant factor.

\begin{lemma}\label{lem:q_frob}
Let $V\in \R^{n\times d}$ have rows $v_1,\ldots,v_n$. There exists a quantum algorithm $\textsc{QFrobNorm}(V)$ that, with probability at least $1-1/\poly(n)$, outputs a number $Z$ with $\|V\|_F^2\leq Z\leq 2\|V\|_F^2$, so $Z=0$ if $V=0$. It uses $\wt O(n^{0.5})$ row queries to $V$ and $\wt O(n^{0.5}d)$ time.
\end{lemma}

\begin{proof}
Let $f(i):=\|v_i\|_2^2$, which costs one row query to $V$ and $O(d)$ time. Let $c\geq 1$ be the desired exponent in the failure probability.

\emph{Step 1: the largest squared row norm.} Applied to $-f$ with ties broken by index, the quantum minimum finding algorithm of~\cite{dh96} returns an index of a maximum of $f$ with probability at least $1/2$, using $O(n^{0.5})$ evaluations of $f$. We run it $\lceil c\log_2 n\rceil+1$ times and keep the returned index with the largest value of $f$. With probability at least $1-n^{-c}/2$, this yields $M:=\max_{i\in [n]}f(i)$. If the value found is $0$, the algorithm outputs $Z:=0$ and stops. This is correct if $V=0$, since then every index is a maximum of $f$ and the value found is $0$ for certain. If $V\neq 0$, then on this success event $M>0$ and $M\leq \|V\|_F^2\leq nM$, and we continue.

\emph{Step 2: amplitude estimation.} Let $\mathcal{A}$ be the unitary that maps $|0\rangle$ to
\begin{align*}
    \frac{1}{\sqrt n}\sum_{i=1}^n|i\rangle\Big(\sqrt{f(i)/M}\,|1\rangle+\sqrt{1-f(i)/M}\,|0\rangle\Big),
\end{align*}
with all work registers returned to $|0\rangle$. It prepares the uniform superposition over $[n]$, computes $f(i)$ with one row query and $O(d)$ operations, rotates the last qubit, and uncomputes $f(i)$ with a second row query. Hence it costs two row queries and $O(d+\log n)$ time. Measuring the last qubit of $\mathcal{A}|0\rangle$ gives $1$ with probability
\begin{align*}
    a:=\frac{1}{n}\sum_{i=1}^n\frac{f(i)}{M}=\frac{\|V\|_F^2}{nM}\in [1/n,1].
\end{align*}
For every integer $t\geq 1$, amplitude estimation (Theorem 12 of~\cite{bhmt02}) uses $O(t)$ applications of $\mathcal{A}$ and $\mathcal{A}^{-1}$, together with $O(t)$ reflections that cost $O(\log n)$ time each, and outputs $\wt a$ with
\begin{align*}
    |\wt a-a|\leq \frac{2\pi\sqrt{a(1-a)}}{t}+\frac{\pi^2}{t^2}
\end{align*}
with probability at least $8/\pi^2$. We take $t=\lceil 8\pi n^{0.5}\rceil$. Since $a\geq 1/n$, we have $1/t\leq a^{1/2}/(8\pi)$, so $2\pi\sqrt a/t\leq a/4$ and $\pi^2/t^2\leq a/64$. Hence $|\wt a-a|<a/3$, that is, $\wt a\in [\frac{2}{3}a,\frac{4}{3}a]$, with probability at least $8/\pi^2>4/5$. We repeat this $R:=\lceil 6(c\ln n+1)\rceil$ times independently and let $\ov a$ be the median of the outputs. The median lies in $[\frac{2}{3}a,\frac{4}{3}a]$ unless at least $R/2$ runs fail. By Hoeffding's inequality, this happens with probability at most $\exp(-2R(\frac{4}{5}-\frac{1}{2})^2)\leq \exp(-0.18R)\leq n^{-c}/2$.

We output $Z:=\frac{3}{2}nM\ov a$. On the two success events, which together fail with probability at most $n^{-c}$, we have $nM\ov a\in [\frac{2}{3}\|V\|_F^2,\frac{4}{3}\|V\|_F^2]$, so $\|V\|_F^2\leq Z\leq 2\|V\|_F^2$. Step 1 uses $O(n^{0.5}\log n)$ evaluations of $f$, and Step 2 uses $O(Rt)=O(n^{0.5}\log n)$ applications of $\mathcal{A}$ and $\mathcal{A}^{-1}$. In total, the algorithm uses $\wt O(n^{0.5})$ row queries and $\wt O(n^{0.5}d)$ time.
\end{proof}

We are ready to state our final result on approximating the value matrix $V$.

\begin{theorem}\label{thm:q_amm}
Let $V\in \R^{n\times d}$ be nonzero with rows $v_1,\ldots,v_n$, and let $s_0\geq 1$ and $1\leq T\leq \poly(n)$ be integers. There exists a quantum algorithm that computes $Z\gets \textsc{QFrobNorm}(V)$ (Lemma~\ref{lem:q_frob}) and then $T$ independent weighted sampling matrices $S^{(1)},\ldots,S^{(T)}$ together with $S^{(1)\top}V,\ldots,S^{(T)\top}V$, where $S^{(t)}\in \R^{n\times s^{(t)}}$ is computed by a run of $\textsc{QSample}$ whose exact sample (Remark~\ref{rem:q_sample_whp}) gives a matrix $\ov S^{(t)}$ keeping each row $i$ independently with probability $p_i=\min\{1,s_0\|v_i\|_2^2/Z\}$. It has the following properties.
\begin{itemize}
    \item With probability at least $1-1/\poly(n)$, we have $\|V\|_F^2\leq Z\leq 2\|V\|_F^2$ and $s^{(t)}\leq \min\{n,O(s_0+\log n)\}$ for every $t\in [T]$, the algorithm uses $\wt O(n^{0.5}+T(ns_0)^{0.5})$ row queries to $V$ and $\wt O((n^{0.5}+T(ns_0)^{0.5})d)$ time, and $S^{(t)}=\ov S^{(t)}$ for every $t\in [T]$.
    \item Conditioned on any value of $Z$ with $\|V\|_F^2\leq Z\leq 2\|V\|_F^2$, the matrices $\ov S^{(1)},\ldots,\ov S^{(T)}$ are independent, and for all $\epsilon,\delta\in (0,1)$ with $s_0\geq C/(\epsilon^2\delta)$, where $C$ is a large enough absolute constant, every fixed $B\in \R^{n\times m}$ and every $t\in [T]$, the matrix $\ov S^{(t)}$ is an $(\epsilon,\delta)$-AMM for $V$ and $B$.
\end{itemize}
In particular, for $T=1$ and $s_0=\lceil 4C/\epsilon^2\rceil$, the algorithm uses $\wt O(\epsilon^{-1}n^{0.5})$ row queries to $V$ and $\wt O(\epsilon^{-1}n^{0.5}d)$ time, $S^{(1)}$ has $O(\epsilon^{-2}+\log n)$ columns, and for every fixed $B$, the bound $\|V^\top S^{(1)}S^{(1)\top}B-V^\top B\|_F\leq \epsilon\|V\|_F\|B\|_F$ holds with probability at least $3/4-1/\poly(n)$.
\end{theorem}

\begin{proof}
By Lemma~\ref{lem:q_frob}, $Z$ satisfies $\|V\|_F^2\leq Z\leq 2\|V\|_F^2$ with probability at least $1-1/\poly(n)$, using $\wt O(n^{0.5})$ row queries and $\wt O(n^{0.5}d)$ time. Assume this event. Given $Z$, each $p_i$ costs one row query to $V$ and $O(d)$ time, and
\begin{align*}
    \sum_{i=1}^n p_i\leq \min\Big\{n,\frac{s_0}{Z}\sum_{i=1}^n\|v_i\|_2^2\Big\}\leq \min\{n,s_0\}.
\end{align*}
For each $t\in [T]$, we run $\textsc{QSample}(p)$ as in Remark~\ref{rem:q_sample_whp} with fresh randomness, and $\ov S^{(t)}$ is the weighted sampling matrix of the exact sample of this run. By Remark~\ref{rem:q_sample_whp}, for any constant $c\geq 1$, with probability at least $1-n^{-c}$ the run returns its exact sample, so that $S^{(t)}=\ov S^{(t)}$, and uses $\wt O((n\max\{1,\sum_i p_i\})^{0.5})\leq \wt O((ns_0)^{0.5})$ row queries and $\wt O((ns_0)^{0.5}d)$ time, where we used $s_0\geq 1$. The number of rows kept by $\ov S^{(t)}$ is at most $n$, and it is a sum of independent indicators with mean at most $s_0$. A Chernoff bound, in the form $\Pr[X\geq x]\leq 2^{-x}$ for $x\geq 6\E[X]$, gives, for any constant $c\geq 1$, that this number is at most $\max\{6s_0,c\log_2 n\}$ with probability at least $1-n^{-c}$. When $S^{(t)}=\ov S^{(t)}$, the same bound holds for $s^{(t)}$. A union bound covers all $T\leq \poly(n)$ runs. Reading the kept rows to form $S^{(t)\top}V$ takes $s^{(t)}$ row queries and $O(s^{(t)}d)$ time, and $s^{(t)}\leq \min\{n,O(s_0+\log n)\}\leq O((ns_0)^{0.5}+\log n)$. This proves the first item.

By Remark~\ref{rem:q_sample_whp}, conditioned on $Z$, on everything computed before the $t$-th run and on $\ov S^{(1)},\ldots,\ov S^{(t-1)}$, $\ov S^{(t)}$ keeps each row $i$ independently with probability $p_i$, and these probabilities depend only on $Z$. Hence, conditioned on $Z$, the matrices $\ov S^{(1)},\ldots,\ov S^{(T)}$ are independent. Since $Z\leq 2\|V\|_F^2$, we have $p_i\geq \min\{1,s_0\|v_i\|_2^2/(2\|V\|_F^2)\}$. Hence Lemma~\ref{lem:amm} with $A=V$, $s=s_0$ and $\gamma=2$ shows that each $\ov S^{(t)}$ is an $(\epsilon,\delta)$-AMM for $V$ and every fixed $B$ whenever $s_0\geq 2/(\epsilon^2\delta)$, so $C=2$ suffices. This proves the second item. The last claim is the case $T=1$ and $\delta=1/4$, combined with the first item, which fails with probability $1/\poly(n)$ and on whose event $S^{(1)}=\ov S^{(1)}$.
\end{proof}

\begin{remark}\label{rem:no_alpha}
Since $\alpha(V)\geq 1$ (Lemma~\ref{lem:row_distortion_bounds}), the cost $\wt O(\epsilon^{-1}n^{0.5})$ in Theorem~\ref{thm:q_amm} is at most the cost $\wt O(\epsilon^{-1}n^{0.5}\alpha^{0.5})$ of drawing $O(\epsilon^{-2}\alpha)$ rows of $V$ by leverage scores with $\textsc{QSample}$, even if the leverage scores were known in advance. Estimating them through a spectral approximation of $V$ (Lemma~\ref{lem:q_spectral_approx}) would add $\wt O((nd)^{0.5})$ row queries and $\wt O(n^{0.5}d^{1.5}+d^\omega)$ time. Sampling by squared row norms needs neither.
\end{remark}
\section{Put Things Together}
\label{sec:result}
We are now ready to state our final algorithm and its guarantee. Recall that, we define $D={\rm diag}(\exp(QK^\top){\bf 1}_n)$ and $D'={\rm diag}(\exp(KQ^\top){\bf 1}_n)$. We use $\wt D, \wt D'$ to denote their approximations.

We prove a simple inequality that quantifies the perturbation on the inverse.
\begin{lemma}\label{lem:inverse_perturb}
Let $C, D\in \R^{n\times n}$, if $D$ is nonsingular and $\|C-D\|\leq \epsilon$, and $\|D^{-1}\|< 1/\epsilon$, then $C$ is also nonsingular and $\|C^{-1}\|\leq \frac{\|D^{-1}\|}{1-\epsilon\cdot \|D^{-1}\|}$.
\end{lemma}

\begin{proof}
We will make use of Neumann series, which states that for $\|A\|<1$, $(I-A)^{-1}$ admits the expansion
\begin{align*}
    (I-A)^{-1} = & ~ \sum_{k=0}^\infty A^k,
\end{align*}
this leads to a bound on the norm:
\begin{align}\label{eq:neumann_series}
    \|(I-A)^{-1}\| = & ~ \|\sum_{k=0}^\infty A^k\| \notag\\
    \leq & ~ \sum_{k=0}^\infty \|A^k\| \notag\\
    \leq & ~ \sum_{k=0}^\infty \|A\|^k \notag\\
    = & ~ \frac{1}{1-\|A\|},
\end{align}
now, to prove our desired bound, we write $C=D+E$ where $E$ is the perturbation, then $C=D+E=D(I+D^{-1}E)$, and we will apply Eq.~\eqref{eq:neumann_series} to $-D^{-1}E$: 
\begin{align*}
    \|D^{-1}E \| \leq & ~ \|D^{-1}\|\cdot\|E\| \\
    = & ~ \|D^{-1}\|\cdot \|C-D\| \\
    < & ~ 1/\epsilon\cdot \epsilon \\
    = & ~ 1,
\end{align*}
therefore
\begin{align*}
    \|C^{-1}\| = & ~ \|(I+D^{-1}E)^{-1}D^{-1}\| \\
    \leq & ~ \|D^{-1}\|\cdot \|(I+D^{-1}E)^{-1}\| \\
    \leq & ~ \frac{\|D^{-1}\|}{1-\|D^{-1}E\|} \\
    \leq & ~ \frac{\|D^{-1}\|}{1-\|E\|\cdot\|D^{-1}\|} \\
    \leq & ~ \frac{\|D^{-1}\|}{1-\epsilon\cdot \|D^{-1}\|},
\end{align*}
this completes the proof.
\end{proof}

\begin{theorem}[Formal version of Theorem~\ref{thm:main_informal}]
\label{thm:main_formal}
Let $Q, K, V\in \R^{n\times d}$ be the query, key and value matrices for attention, let $\epsilon\in (0,1)$ and $\lambda>0$. Let $E\in \R^{2n\times 2n}$ be the exponential kernel matrix with the dataset $Q\cup K$, and let $s_\lambda$ be the statistical dimension of $E$ (Definition~\ref{def:stat_dim}). Assume that $\|D^{-1}\|<\frac{1}{\epsilon\|A\|+\lambda\sqrt n}$ and let $\beta=\frac{1}{1-(\epsilon \|A\|+\lambda \sqrt n)\|D^{-1}\|}$. Assume also that the matrix $U_2$ formed in Algorithm~\ref{alg:norm_est} has full column rank. There exists a quantum data structure (Algorithm~\ref{alg:q_attention}) that preprocesses $Q, K, V$ only through row queries to these matrices and draws $T=O(\log n)$ sampling matrices $S_V^{(1)},\ldots,S_V^{(T)}$ of $V$. For any $i\in [n]$, it outputs a vector $\wt r_i\in \R^d$ where
\begin{align*}
    \wt r_i = & ~ e_i^\top \wt D^{-1} \wt A S_V^{(t_i)}S_V^{(t_i)\top}V
\end{align*}
for an index $t_i\in [T]$ that it selects at query time. With probability at least $1-1/\poly(n)$, the matrix $\wt R\in \R^{n\times d}$ whose $i$-th row is $\wt r_i$ satisfies
\begin{align*}
    \|\wt R -  D^{-1} A V \|_F
    \leq & ~ 3\beta\cdot \|D^{-1}\|\cdot (\epsilon\|A\|_F+\lambda\sqrt n)\cdot \|V\|_F.
\end{align*}
In particular, if $\lambda\sqrt n\leq \epsilon\|A\|$, then the error is $O(\epsilon\cdot \beta\cdot \|D^{-1}\|\cdot \|A\|_F\cdot \|V\|_F)$.

On the same event, the algorithm also has the following runtime specification:
\begin{itemize}
    \item Preprocesses in $\wt O(\epsilon^{-1}n^{0.5}s_\lambda)$ row queries to $Q, K$ and $\wt O(\epsilon^{-1}n^{0.5})$ row queries to $V$, and $\wt O(\epsilon^{-1}n^{0.5}(s_\lambda^{3}+s_\lambda^{2}d+d)+s_\lambda^\omega+\epsilon^{-2}s_\lambda d)$ time.
    \item For any $i\in [n]$, it outputs $\wt r_i$ in $\wt O(s_\lambda^2+s_\lambda d)$ time.
\end{itemize}
\end{theorem}

\begin{proof}
If $V=0$, then $\textsc{QFrobNorm}$ outputs $Z=0$ (Lemma~\ref{lem:q_frob}), and the data structure returns ${\bf 0}_d$ for every query, which is exact. We assume $V\neq 0$ from now on. Let $\delta=1/\poly(n)$ be the parameter of Algorithm~\ref{alg:q_attention}, so that $T=O(\log n)$ and $s_0=O(\epsilon^{-2})$. Let $\wt A$ be the off-diagonal block of the Nystr{\"o}m approximation $ES_E(S_E^\top ES_E)^\dagger S_E^\top E$, and let $\wh A=U_1U_2^\top$ be the corresponding block for the sampling matrix computed inside $\text{QRN}.\textsc{Preprocess}$. We use the following events.
\begin{itemize}
    \item $\mathcal{G}_1$: both Nystr{\"o}m sampling matrices satisfy the guarantee of Theorem~\ref{thm:q_nystrom}, have $\wt O(s_\lambda)$ columns, and are computed within the running time of that theorem. By Corollary~\ref{cor:attention_approx}, then $\|\wt A-A\|\leq \lambda$, $\|\wt A-A\|_F\leq \lambda\sqrt n$ and $\|\wh A-A\|\leq \lambda$.
    \item $\mathcal{G}_2$: the vector $\wt \mu$ computed by $\text{QRN}.\textsc{Preprocess}$, which runs with accuracy $\epsilon/2$, satisfies $\|\wt \mu-U_2^\top {\bf 1}_n\|_{(U_2^\top U_2)^{-1}}\leq \epsilon/2$ (Lemma~\ref{lem:q_matvec}).
    \item $\mathcal{G}_3$: the first item of Theorem~\ref{thm:q_amm} holds for $Z$ and $S_V^{(1)},\ldots,S_V^{(T)}$. In particular, $S_V^{(t)}=\ov S_V^{(t)}$ for every $t\in [T]$, where $\ov S_V^{(t)}$ is the matrix $\ov S^{(t)}$ of that theorem.
\end{itemize}
Each of these events fails with probability at most $1/\poly(n)$, and we assume all three.

\emph{Normalization.} The output of $\text{QRN}.\textsc{Query}(i)$ is $b_i=\langle (U_1)_{i,*},\wt \mu\rangle$, so the vector $b=(b_1,\ldots,b_n)^\top$ equals $U_1\wt \mu$, and $\wt D={\rm diag}(b)$. The hypothesis gives $\min_i D_{i,i}=1/\|D^{-1}\|>\epsilon\|A\|+\lambda\sqrt n\geq\lambda\sqrt n$, and $D_{i,i}=\|a_i\|_1\leq \sqrt n\|a_i\|_2\leq \sqrt n\|A\|$, so $\lambda<\|A\|$. On $\mathcal{G}_2$, Corollary~\ref{cor:mat_vec_norm} with $v={\bf 1}_n$ gives $\|U_1\wt \mu-\wh A{\bf 1}_n\|_2\leq \frac{\epsilon}{2}\|\wh A\|\leq \frac{\epsilon}{2}(\|A\|+\lambda)$. On $\mathcal{G}_1$, we have $\|(\wh A-A){\bf 1}_n\|_2\leq \|\wh A-A\|\cdot \|{\bf 1}_n\|_2\leq \lambda\sqrt n$. Since the diagonal of $D$ is $A{\bf 1}_n$, this gives
\begin{align}\label{eq:row_sum_error}
    \|b-A{\bf 1}_n\|_2\leq \frac{\epsilon}{2}(\|A\|+\lambda)+\lambda\sqrt n<\epsilon\|A\|+\lambda\sqrt n,
\end{align}
where the last step uses $\lambda<\|A\|$. The matrix $\wt D-D$ is diagonal with entries $b_i-D_{i,i}$, so $\|\wt D-D\|\leq \|b-A{\bf 1}_n\|_2<\epsilon\|A\|+\lambda\sqrt n$. Since $\|D^{-1}\|<\frac{1}{\epsilon\|A\|+\lambda\sqrt n}$, Lemma~\ref{lem:inverse_perturb}, applied with $C=\wt D$ and with $\epsilon\|A\|+\lambda\sqrt n$ in place of $\epsilon$, shows that $\wt D$ is nonsingular and
\begin{align*}
    \|\wt D^{-1}\| \leq & ~ \frac{\|D^{-1}\|}{1-(\epsilon \|A\|+\lambda \sqrt n)\|D^{-1}\|}=\beta\cdot \|D^{-1}\|.
\end{align*}

\emph{Decomposition of the error.} Let $W:=\wt D^{-1}\wt A$. Then
\begin{align*}
    \wt R-D^{-1}AV=\underbrace{(\wt R-WV)}_{\Delta_1}+\underbrace{\wt D^{-1}(\wt A-A)V}_{\Delta_2}+\underbrace{(\wt D^{-1}-D^{-1})AV}_{\Delta_3}.
\end{align*}

\emph{The sampling error $\Delta_1$.} For each $t\in [T]$, $\textsc{Query}(i)$ computes $r^{(t)}=\wt L^{(t)\top}u_i/b_i$. Since $u_i^\top=e_i^\top ES_E$, $\wt M=S_E^\top ES_E$ and $\wt H^{(t)}=S_E^\top E\begin{bmatrix}{\bf 0}_{n\times s_V^{(t)}}\\ S_V^{(t)}\end{bmatrix}$, we have $u_i^\top \wt M^\dagger \wt H^{(t)}=e_i^\top \wt A S_V^{(t)}$, so $r^{(t)}$ is the transpose of $e_i^\top WS_V^{(t)}S_V^{(t)\top}V$. The selection rule of $\textsc{Query}$ is the one of Lemma~\ref{lem:amm_median}, so $\wt R$ is the matrix of that lemma for $W$ and $S_V^{(1)},\ldots,S_V^{(T)}$. The matrix $W$ depends only on $S_E$ and on the randomness of QRN, so it is independent of $Z$ and of the samples of $V$. Conditioned on any $Z$ with $\|V\|_F^2\leq Z\leq 2\|V\|_F^2$, Theorem~\ref{thm:q_amm} shows that $\ov S_V^{(1)},\ldots,\ov S_V^{(T)}$ are independent and that each of them is an $(\epsilon/3,1/4)$-AMM for $V$ and every fixed $B$, because $s_0\geq 72/\epsilon^2=2/((\epsilon/3)^2\cdot \frac{1}{4})$ and $C=2$ suffices in Theorem~\ref{thm:q_amm}. Taking $B$ to be a fixed vector $w\in \R^n$ gives the hypothesis of Lemma~\ref{lem:amm_median} for these matrices, and $T\geq 8\ln(n/\delta)$ meets its condition on $T$ with $C=8$. Hence, except with probability $\delta$, the matrix of that lemma for $W$ and $\ov S_V^{(1)},\ldots,\ov S_V^{(T)}$ is within $\epsilon\|W\|_F\|V\|_F$ of $WV$ in Frobenius norm. On $\mathcal{G}_3$ the computed matrices equal the exact samples, so this matrix is $\wt R$ and
\begin{align*}
    \|\Delta_1\|_F\leq \epsilon\|W\|_F\|V\|_F\leq \epsilon\|\wt D^{-1}\|\cdot \|\wt A\|_F\cdot \|V\|_F\leq \epsilon\beta\|D^{-1}\|(\|A\|_F+\lambda\sqrt n)\|V\|_F,
\end{align*}
where the last step uses $\|\wt A\|_F\leq \|A\|_F+\|\wt A-A\|_F\leq \|A\|_F+\lambda\sqrt n$.

\emph{The Nystr{\"o}m error $\Delta_2$.} Since $\|XYV\|_F\leq \|X\|\cdot \|Y\|\cdot \|V\|_F$,
\begin{align*}
    \|\Delta_2\|_F\leq \|\wt D^{-1}\|\cdot \|\wt A-A\|\cdot \|V\|_F\leq \beta\|D^{-1}\|\lambda\|V\|_F.
\end{align*}

\emph{The normalization error $\Delta_3$.} Since $\wt D$ and $D$ are diagonal, $\Delta_3=(\wt D^{-1}-D^{-1})D\cdot D^{-1}AV$, and the $i$-th row of $\Delta_3$ is $\frac{D_{i,i}-b_i}{b_i}\cdot e_i^\top D^{-1}AV$. The entries $A_{i,j}/D_{i,i}$ of $e_i^\top D^{-1}A$ are nonnegative and sum to $1$, so $e_i^\top D^{-1}AV=\sum_{j=1}^n\frac{A_{i,j}}{D_{i,i}}v_j^\top$ is a convex combination of the rows of $V$. Hence $\|e_i^\top D^{-1}AV\|_2\leq \max_j\|v_j\|_2\leq \|V\|_F$. Using $\max_i 1/|b_i|=\|\wt D^{-1}\|$, we get
\begin{align*}
    \|\Delta_3\|_F^2=\sum_{i=1}^n\frac{(D_{i,i}-b_i)^2}{b_i^2}\|e_i^\top D^{-1}AV\|_2^2\leq \|\wt D^{-1}\|^2\cdot \|b-A{\bf 1}_n\|_2^2\cdot \|V\|_F^2,
\end{align*}
and by Eq.~\eqref{eq:row_sum_error}, $\|\Delta_3\|_F\leq \beta\|D^{-1}\|(\epsilon\|A\|+\lambda\sqrt n)\|V\|_F$.

\emph{Total error.} Adding the three bounds and using $\epsilon<1$, $\|A\|\leq \|A\|_F$ and $\lambda\leq \lambda\sqrt n$,
\begin{align*}
    \|\wt R-D^{-1}AV\|_F\leq & ~ \beta\|D^{-1}\|\big(\epsilon\|A\|_F+\epsilon\lambda\sqrt n+\lambda+\epsilon\|A\|+\lambda\sqrt n\big)\|V\|_F\\
    \leq & ~ \beta\|D^{-1}\|(2\epsilon\|A\|_F+3\lambda\sqrt n)\|V\|_F\\
    \leq & ~ 3\beta\cdot \|D^{-1}\|\cdot (\epsilon\|A\|_F+\lambda\sqrt n)\cdot \|V\|_F.
\end{align*}
By a union bound over $\mathcal{G}_1,\mathcal{G}_2,\mathcal{G}_3$ and the failure probability $\delta$ of the bound on $\Delta_1$, this holds with probability at least $1-1/\poly(n)$. All these events concern only the preprocessing, so the bound covers the answers to all queries at once. If $\lambda\sqrt n\leq \epsilon\|A\|$, the right-hand side is at most $6\epsilon\beta\|D^{-1}\|\|A\|_F\|V\|_F$.

\emph{Running time.} On $\mathcal{G}_1$ and $\mathcal{G}_3$, let $s=\wt O(s_\lambda)$ bound $s_E$ and the number of columns of the Nystr{\"o}m sampling matrix of QRN, and note that $s_V^{(t)}\leq \min\{n,O(\epsilon^{-2}+\log n)\}$ for every $t$. The steps of $\textsc{Preprocess}$ cost as follows.
\begin{itemize}
    \item $\text{QRN}.\textsc{Preprocess}$: $\wt O(n^{0.5}s/\epsilon)$ row queries to $Q,K$ and $\wt O(n^{0.5}s^2(s+d)/\epsilon+s^\omega)$ time (Theorem~\ref{thm:q_norm}).
    \item $Z$, the samples $S_V^{(t)}$ and the matrices $\wt V^{(t)}$: with $s_0=O(\epsilon^{-2})$ and $T=O(\log n)$, Theorem~\ref{thm:q_amm} gives $\wt O(\epsilon^{-1}n^{0.5})$ row queries to $V$ and $\wt O(\epsilon^{-1}n^{0.5}d)$ time.
    \item $S_E$: $\wt O(n^{0.5}s^{0.5})$ row queries to $Q,K$ and $\wt O(n^{0.5}s^{1.5}(s+d)+s^\omega)$ time (Theorem~\ref{thm:q_nystrom} and Corollary~\ref{cor:approx_exp}, whose factor $\log(n/\delta)$ is $O(\log n)$ because Algorithm~\ref{alg:q_attention} calls them with failure probability $1/\poly(n)$).
    \item $\wt M$ and $\wt M^\dagger$: $O(s^2d+s^\omega)$ time.
    \item For each $t\in [T]$: $\wt H^{(t)}$ needs the keys $k_j$ for the sampled indices $j$, which costs $s_V^{(t)}\leq O((n(\epsilon^{-2}+\log n))^{0.5})$ row queries to $K$, and $O(s\cdot s_V^{(t)}d)$ time. Then $\wt H^{(t)}\wt V^{(t)}$ takes $O(s\cdot s_V^{(t)}d)$ time and $\wt M^\dagger(\wt H^{(t)}\wt V^{(t)})$ takes $O(s^2d)$ time. Over the $T=O(\log n)$ values of $t$, this is $\wt O(\epsilon^{-1}n^{0.5})$ row queries and $\wt O(\epsilon^{-2}s_\lambda d+s_\lambda^2d)$ time.
\end{itemize}
Since $n^{0.5}/\epsilon\geq 1$, the preprocessing uses $\wt O(\epsilon^{-1}n^{0.5}s_\lambda)$ row queries to $Q,K$, $\wt O(\epsilon^{-1}n^{0.5})$ row queries to $V$, and $\wt O(\epsilon^{-1}n^{0.5}(s_\lambda^3+s_\lambda^2d+d)+s_\lambda^\omega+\epsilon^{-2}s_\lambda d)$ time. A query computes $b_i$ in $O(s(s+d))$ time (Theorem~\ref{thm:q_norm}), $u_i$ in $O(sd)$ time, the $T$ candidates in $O(Tsd)$ time, and the $T^2$ distances and the selection in $O(T^2d+T^2\log T)$ time, which is $\wt O(s_\lambda^2+s_\lambda d)$ in total.
\end{proof}
\section{Empirical Verifications on Parameters}
\label{sec:exp}
In this section, we empirically verify the assumptions on the parameters. In particular, we focus on the following metrics:
\begin{itemize}
    \item $\|D^{-1}\|\leq \frac{1}{\epsilon\|A\|+\lambda\sqrt{n}}$, we specifically check that what is the maximum possible $\epsilon$ so that $\|D^{-1}\|\leq \frac{1}{\epsilon\|A\|}$.
    \item $\frac{\|A\|_F}{\|A\|}$, this is important as our error guarantee is in terms of Frobenius norm rather than the more typical spectral norm~\citep{zhdk23,hjk+24}, we verify that this ratio is small.
    \item $\frac{\|V\|_F}{\|V\|}$, this is similar to the above test, we verify that this ratio is close to $\sqrt{d}$.
    \item $\frac{d}{{\rm srank}(V)}$, this quantity serves as an upper bound of $\alpha(V)$, we verify that this quantity is a small constant rather than the upper bound $d$. The sampler of Theorem~\ref{thm:q_amm} does not depend on the row distortion, which enters only leverage score sampling of $V$ (Lemma~\ref{lem:amm}).
    \item $\frac{\|A\|_\infty}{\|A\|}$, in our error analysis, we have to pay an extra $\sqrt{n}$ factor when converting the spectral norm to matrix infinity norm, we empirically show that this ratio is a small constant rather than the $\sqrt{n}$ scaling.
\end{itemize}
To conduct our experiment, we use the \texttt{OLMo2-1B} and \texttt{OLMo2-7B} models, in particular their \texttt{stage1} pretraining checkpoints~\citep{wsg+25}. We list the model architecture in the following.

\begin{table}[!ht]
    \centering
    \begin{tabular}{|l|l|l|l|l|}
    \hline
         &  {\bf Sequence length $n$} & {\bf Value dimension $d$} & {\bf Number of layers $L$} & {\bf Number of heads $H$} \\ \hline
      \texttt{OLMo2-1B}  & 4096 & 128 & 16 & 16  \\ \hline
      \texttt{OLMo2-7B} & 4096 & 128 & 32 & 32 \\ \hline
    \end{tabular}
    \caption{Model architecture for \texttt{OLMo2-1B} and \texttt{OLMo2-7B}.}
    \label{tab:model_architecture}
\end{table}
We compute the corresponding attention modules $D, A, V$ using the pretraining datasets for these models, with batch size 2 and 16 batches. We then compute the statistics for each head and each layer, then aggregate the statistics over all layers. We report the mean of these statistics.

\begin{table}[!ht]
    \centering
    \begin{tabular}{|l|l|l|l|l|l|}
    \hline
         &  $\epsilon_{\max}$ & $\frac{\|A\|_F}{\|A\|}$ & $\frac{\|V\|_F}{\|V\|}$ & $\frac{d}{{\rm srank}(V)}$ & $\frac{\|A\|_\infty}{\|A\|}$ \\ \hline
       \texttt{OLMo2-1B}  & 0.1708 & 1.3769 & 11.3137 & 1.7126 & 2.7439 \\ \hline
       \texttt{OLMo2-7B} & 0.1685 & 1.3586 & 11.3137 & 2.1345 & 2.7439 \\ \hline
    \end{tabular}
    \caption{Mean statistics across all heads and all layers. $\epsilon_{\max}$ is the maximum $\epsilon$ such that $\|D^{-1}\|\leq \frac{1}{\epsilon\|A\|}$.}
    \label{tab:stats}
\end{table}

Through these verifications, we make the following preliminary observations:
\begin{itemize}
    \item To satisfy the $\|D^{-1}\|\leq \frac{1}{\epsilon\|A\|}$ assumption, it is enough to pick $\epsilon\leq 0.17$, which is larger than common choice of $\epsilon\approx 0.1$. This gives us enough room to tune the parameter $\epsilon$ to achieve a good balance between efficiency and accuracy.
    \item The ratio $\frac{\|A\|_F}{\|A\|}$ is a constant smaller than 2, much smaller than the worst case $\sqrt{n}$ predicted by the theory (recall that $n=4096$ and $\sqrt{n}=64$). This suggests that we don't need to scale down $\epsilon$ by a factor of $\sqrt{n}$ to recover the spectral norm error guarantee.
    \item The ratio $\frac{\|V\|_F}{\|V\|}\approx 11$ is roughly $\sqrt{d}$ as $d=128$, this confirms the theory, but it does not impair the sublinear scaling in $n$ of our algorithm: we could simply scale down $\epsilon$ by a factor of $\sqrt{d}$ to absorb this blowup and increase the runtime by a factor of $\sqrt{d}$.
    \item The ratio $\frac{d}{{\rm srank}(V)}$ is a small constant, recall that ${\rm srank}(V)$ can be as small as 1, causing the ratio to be $d$, our experiment shows that $\alpha(V)$ is close to a small constant rather than $d$, so even leverage score sampling of $V$ would lose only a small constant factor.
    \item The ratio $\frac{\|A\|_\infty}{\|A\|}$ is a constant smaller than 3, we check this quantity as in proving the approximation guarantee for $\wt D$, we make use of the fact that $\|A\|_\infty\leq \sqrt{n}\cdot \|A\|$, this again shows that instead of the worst case $\sqrt{n}$ scaling, this distortion is only by a constant factor, implying that the $\lambda\sqrt{n}$ term in the error of each normalization factor (Theorem~\ref{thm:q_norm}), which enters the hypothesis on $\|D^{-1}\|$, is more likely $O(\lambda)$ in practice. This greatly enlarges the range of choice for $\lambda$ to achieve better speedup.
\end{itemize}
\section{Bit Complexity of Our Algorithm}
\label{sec:bit_complexity}
In this section, we give a preliminary analysis on the bit complexity of our algorithm, in particular the bit complexity of matrix inversion operation. We will make use of the following standard algorithm for backward stable matrix inversion.
\begin{lemma}[\cite{h02,hh21}]
Let $A\in \R^{s\times s}$ be nonsingular, there exists an algorithm that computes $B^{-1}$ such that 
\begin{align*}
    \|B-A\|\leq & ~ \delta \cdot s^c\cdot \kappa(A)^{C\log s}\cdot \|A\|,
\end{align*}
for absolute constant $c, C>0$ with bit complexity $O(s^3\cdot M(b))$ where $b=O(\log(\kappa(A))+\log(1/\delta))$ and $M(b)=O(b\log b)$.
\end{lemma}

We note that the backward stable error guarantee is exactly what has been analyzed in~\cite{gsyz24}:

\begin{lemma}[Lemma G.3 in~\cite{gsyz24}]
Let $A, B\in \R^{s\times s}$ be matrices such that $\|A-B\|\leq \delta$, then 
\begin{align*}
    |\tau_i(A)-\tau_i(B)|\leq & ~ \delta\cdot \kappa^{2.5}(A).
\end{align*}
\end{lemma}
This means that by setting $\delta=1/\poly(\kappa(A))$, we can approximate the leverage scores well, and the number of bits $b=O(\log(\kappa(A)))$. Note that we apply matrix inversions for two type of matrices:
\begin{itemize}
    \item $S^\top ES+\lambda I$, where $S$ is the ridge leverage score sampling matrix for $E$ (inside Algorithm~\ref{alg:q_nystrom}). The pseudo-inverses $(S^\top ES)^{\dagger/2}$ and $(S_E^\top ES_E)^\dagger$ of Algorithms~\ref{alg:norm_est} and~\ref{alg:q_attention} additionally require a lower bound on the smallest nonzero eigenvalue of $S^\top ES$, which we do not analyze.
    \item the $s\times s$ spectral approximation of $U_2^\top U_2$ whose inverse square root $\textsc{QMatVec}$ computes (Lemma~\ref{lem:q_matvec}), which needs a lower bound on the smallest eigenvalue of $U_2^\top U_2$ that we do not analyze either.
\end{itemize}
Sampling the rows of $V$ (Theorem~\ref{thm:q_amm}) inverts no matrix, so the condition number of $V$ plays no role. We analyze the first case. 

To analyze $\kappa(S^\top ES+\lambda I)$, we upper bound the spectral norm and lower bound the smallest eigenvalue. First, observe that $S^\top ES\succeq 0$, so trivially we have $S^\top ES+\lambda I\succeq \lambda I$, thus the smallest eigenvalue is at least $\lambda$. To bound $\|S^\top ES+\lambda I\|$, we note that
\begin{align*}
    \|S^\top ES+\lambda I\| \leq & ~ \|S^\top ES\|+\lambda \\
    \leq & ~ \|S\|^2\cdot\|E\|+\lambda,
\end{align*}
we bound the two spectral norms respectively. For $\|S\|^2$, we bound it probabilistically: let $c_i=\begin{cases}
    1, & \text{if $i$ is sampled with probability $p_i$} \\
    0, & \text{otherwise}
\end{cases}$, and consider the matrix $SS^\top$, note that by definition, $SS^\top$ is a diagonal matrix with 
\begin{align*}
    (SS^\top)_{i,i} = & ~ \frac{c_i}{p_i},
\end{align*}
note that as $c_i$ is a Bernoulli random variable with probability $p_i$, we have $\E[c_i]=p_i$ hence $\E\left[\frac{c_i}{p_i}\right]=1$, and
\begin{align*}
    \E[\|S\|^2] = & ~ \E\left[\max_{i\in [n]} \frac{c_i}{p_i}\right] \\
    \leq & ~ \E\left[\sum_{i=1}^n \frac{c_i}{p_i}\right] \\
    = & ~ \sum_{i=1}^n \E\left[\frac{c_i}{p_i}\right] \\
    = & ~ n,
\end{align*}
hence by Markov's inequality, with constant probability (say $0.99$), we have that $\|S\|^2\leq O(n)$. Condition on this event, we analyze $\|E\|$: let $R=\max\{\max_i \|q_i\|_2^2, \max_i \|k_i\|_2^2 \}$, then
\begin{align*}
    \|E\| \leq & ~ \tr[E] \\
    = & ~ \sum_{i=1}^n \exp(\|q_i\|_2^2/\sqrt{d})+\exp(\|k_i\|_2^2/\sqrt{d}) \\
    \leq & ~ 2n \exp(R/\sqrt{d}),
\end{align*}
combining the above, we obtain a final (probabilistic) upper bound on the condition number of $S^\top ES+\lambda I$:
\begin{align*}
    \kappa(S^\top ES+\lambda I) \leq & ~ \frac{\|S\|^2\cdot \|E\|+\lambda}{\lambda} \\
    \leq & ~ 1+\frac{Cn^2 \exp(R/\sqrt{d})}{\lambda},
\end{align*}
this gives the final bound on $\log(\kappa(S^\top ES+\lambda I))$:
\begin{align*}
    \log(\kappa(S^\top ES+\lambda I)) \leq & ~ R/\sqrt{d}+\log(n/\lambda).
\end{align*}
In practice, the data-dependent parameter $R/\sqrt{d}$ is small: for both OLMo2-1B and OLMo2-7B models, these values are 20.1147 and 21.2227 respectively. Hence, for the inversions of $S^\top ES+\lambda I$, the bit complexity is $\wt O(s^3 (d^{-0.5}R+\log(n/\lambda)))$.

Finally, to compute the spectral norms and condition numbers required by the algorithms, we could use the algorithms in~\cite{mms18,s25_soda,s25_icalp}.

\section*{Acknowledgment}

We would like to thank anonymous for very helpful discussions, and Ruizhe Zhang for answering our questions on the QRAM model. Lichen Zhang is supported by a Mathworks Fellowship and a Simons Dissertation Fellowship in Mathematics. The corrections in this version were found and written with the help of Claude Opus 5.5. The authors have verified the corrections and take full responsibility for the content of this paper.
\bibliographystyle{alpha}
\bibliography{ref}
\else
\section*{Ethics Statement}

Our work is a theoretical quantum framework to approximate the attention module in sublinear time. We don't foresee any potential ethics concerns.

\section*{Reproducibility Statement}

We include all the proofs in the appendix. For proofs of the exponential kernel, see Section~\ref{sec:kernel}, for proofs of estimating the normalization factor, see Section~\ref{sec:norm}. For proofs of approximate matrix multiplication by row sampling, see Section~\ref{sec:leverage}. The final conclusion is proved in Section~\ref{sec:result}. We provide empirical verifications on the assumptions of the parameters in Section~\ref{sec:exp}, and a preliminary bit complexity analysis of our algorithm in Section~\ref{sec:bit_complexity}.

\section*{Acknowledgment}

We would like to thank anonymous for very helpful discussions, and Ruizhe Zhang for answering our questions on the QRAM model. Lichen Zhang is supported by a Mathworks Fellowship and a Simons Dissertation Fellowship in Mathematics. The corrections in this version were found and written with the help of Claude Opus 5.5. The authors have verified the corrections and take full responsibility for the content of this paper.
\bibliographystyle{iclr2026_conference}
\bibliography{ref}

\newcommand{\etalchar}[1]{$^{#1}$}
\begin{thebibliography}{vAGGdW20b}

\bibitem[AAJ{\etalchar{+}}24]{aaj+24}
Muhammad Adnan, Akhil Arunkumar, Gaurav Jain, Prashant~J. Nair, Ilya
  Soloveychik, and Purushotham Kamath.
\newblock Keyformer: Kv cache reduction through key tokens selection for
  efficient generative inference.
\newblock In {\em Proceedings of the 7th Conference on Machine Learning and
  Systems (MLSys)}, 2024.

\bibitem[ADW22]{aw22}
Simon Apers and Ronald De~Wolf.
\newblock Quantum speedup for graph sparsification, cut approximation, and
  laplacian solving.
\newblock {\em SIAM Journal on Computing}, 51(6):1703--1742, 2022.

\bibitem[AG23]{ag23}
Simon Apers and Sander Gribling.
\newblock Quantum speedups for linear programming via interior point methods.
\newblock {\em arXiv preprint arXiv:2311.03215}, 2023.

\bibitem[AM15]{am15}
Ahmed~El Alaoui and Michael~W. Mahoney.
\newblock Fast randomized kernel ridge regression with statistical guarantees.
\newblock In {\em Advances in Neural Information Processing Systems 28 (NeurIPS
  2015)}, pages 775--783, 2015.

\bibitem[AOA{\etalchar{+}}20]{aoa+20}
Joshua Ainslie, Santiago Ontanon, Chris Alberti, Vaclav Cvicek, Zachary Fisher,
  Philip Pham, Anirudh Ravula, Sumit Sanghai, Qifan Wang, and Li~Yang.
\newblock {ETC}: Encoding long and structured inputs in transformers.
\newblock In {\em Proceedings of the 2020 Conference on Empirical Methods in
  Natural Language Processing (EMNLP)}, pages 268--284. Association for
  Computational Linguistics, 2020.

\bibitem[AS23]{as23}
Josh Alman and Zhao Song.
\newblock Fast attention requires bounded entries.
\newblock {\em Advances in Neural Information Processing Systems},
  36:63117--63135, 2023.

\bibitem[BCE{\etalchar{+}}23]{bce+23}
Sébastien Bubeck, Varun Chandrasekaran, Ronen Eldan, Johannes Gehrke, Eric
  Horvitz, Ece Kamar, Peter Lee, Yin~Tat Lee, Yuanzhi Li, Scott Lundberg,
  Harsha Nori, Hamid Palangi, Marco~Tulio Ribeiro, and Yi~Zhang.
\newblock Sparks of artificial general intelligence: Early experiments with
  gpt-4, 2023.

\bibitem[BFZ{\etalchar{+}}25]{bfz+25}
Payman Behnam, Yaosheng Fu, Ritchie Zhao, Po-An Tsai, Zhiding Yu, and Alexey
  Tumanov.
\newblock Rocketkv: Accelerating long-context llm inference via two-stage kv
  cache compression.
\newblock {\em arXiv preprint arXiv:2502.14051}, 2025.

\bibitem[BHMT02]{bhmt02}
Gilles Brassard, Peter H{\o}yer, Michele Mosca, and Alain Tapp.
\newblock Quantum amplitude amplification and estimation.
\newblock In Samuel J.~Lomonaco Jr. and Howard~E. Brandt, editors, {\em Quantum
  Computation and Information}, volume 305 of {\em Contemporary Mathematics},
  pages 53--74. American Mathematical Society, 2002.
\newblock arXiv:quant-ph/0005055.

\bibitem[BKL{\etalchar{+}}19]{bkl+19}
Fernando G. S.~L. Brand{\~a}o, Amir Kalev, Tongyang Li, Cedric Yen-Yu Lin,
  Krysta~M. Svore, and Xiaodi Wu.
\newblock Quantum sdp solvers: Large speed-ups, optimality, and applications to
  quantum learning.
\newblock In {\em 46th International Colloquium on Automata, Languages, and
  Programming (ICALP 2019)}, volume 132 of {\em Leibniz International
  Proceedings in Informatics (LIPIcs)}, pages 27:1--27:14. Schloss
  Dagstuhl--Leibniz-Zentrum f{\"u}r Informatik, 2019.

\bibitem[BMN{\etalchar{+}}24]{bmn+24}
William Brandon, Mayank Mishra, Aniruddha Nrusimha, Rameswar Panda, and
  Jonathan Ragan-Kelley.
\newblock Reducing transformer key-value cache size with cross-layer attention.
\newblock In {\em Advances in Neural Information Processing Systems (NeurIPS)},
  2024.

\bibitem[BMR{\etalchar{+}}20]{bmr+20}
Tom~B. Brown, Benjamin Mann, Nick Ryder, Melanie Subbiah, Jared Kaplan,
  Prafulla Dhariwal, Arvind Neelakantan, Pranav Shyam, Girish Sastry, Amanda
  Askell, Sandhini Agarwal, Ariel Herbert-Voss, Gretchen Krueger, Tom Henighan,
  Rewon Child, Aditya Ramesh, Daniel~M. Ziegler, Jeffrey Wu, Clemens Winter,
  Christopher Hesse, Mark Chen, Eric Sigler, Mateusz Litwin, Scott Gray,
  Benjamin Chess, Jack Clark, Christopher Berner, Sam McCandlish, Alec Radford,
  Ilya Sutskever, and Dario Amodei.
\newblock Language models are few-shot learners.
\newblock In {\em Advances in Neural Information Processing Systems},
  volume~33, pages 1877--1901. Curran Associates, Inc., 2020.

\bibitem[BPC20]{bpc20}
Iz~Beltagy, Matthew~E. Peters, and Arman Cohan.
\newblock Longformer: The long-document transformer, 2020.

\bibitem[CAM{\etalchar{+}}24]{cam+25}
Timothy Chu, Josh Alman, Gary~L Miller, Shyam Narayanan, Mark Sellke, and Zhao
  Song.
\newblock Metric transforms and low rank representations of kernels for fast
  attention.
\newblock {\em Advances in Neural Information Processing Systems},
  37:47014--47068, 2024.

\bibitem[CBS{\etalchar{+}}15]{cbs+15}
Jan~K Chorowski, Dzmitry Bahdanau, Dmitriy Serdyuk, Kyunghyun Cho, and Yoshua
  Bengio.
\newblock Attention-based models for speech recognition.
\newblock {\em Advances in neural information processing systems}, 28, 2015.

\bibitem[CCLW20]{cclw20}
Shouvanik Chakrabarti, Andrew~M. Childs, Tongyang Li, and Xiaodi Wu.
\newblock Quantum algorithms and lower bounds for convex optimization.
\newblock {\em Quantum}, 4:221, 2020.

\bibitem[CdW23]{cw23}
Yanlin Chen and Ronald de~Wolf.
\newblock Quantum algorithms and lower bounds for linear regression with norm
  constraints.
\newblock In {\em 50th International Colloquium on Automata, Languages, and
  Programming (ICALP 2023)}, volume 261 of {\em Leibniz International
  Proceedings in Informatics (LIPIcs)}, pages 38:1--38:21. Schloss Dagstuhl --
  Leibniz-Zentrum f{\"u}r Informatik, 2023.

\bibitem[CGRS19]{cgrs19}
Rewon Child, Scott Gray, Alec Radford, and Ilya Sutskever.
\newblock Generating long sequences with sparse transformers.
\newblock {\em arXiv preprint arXiv:1904.10509}, 2019.

\bibitem[CHJ22]{chj22}
Arjan Cornelissen, Yassine Hamoudi, and Sofiene Jerbi.
\newblock Near-optimal quantum algorithms for multivariate mean estimation.
\newblock In {\em Proceedings of the 54th Annual ACM SIGACT Symposium on Theory
  of Computing}, STOC 2022, page 33–43, New York, NY, USA, 2022. Association
  for Computing Machinery.

\bibitem[CHL{\etalchar{+}}19]{chl+19}
Shouvanik Chakrabarti, Yiming Huang, Tongyang Li, Soheil Feizi, and Xiaodi Wu.
\newblock Quantum wasserstein generative adversarial networks.
\newblock In {\em Advances in Neural Information Processing Systems 32 (NeurIPS
  2019)}, pages 6768--6779, 2019.

\bibitem[CLD{\etalchar{+}}21]{cld+21}
Krzysztof Choromanski, Valerii Likhosherstov, David Dohan, Xingyou Song,
  Andreea Gane, Tamas Sarlos, Peter Hawkins, Jared Davis, Afroz Mohiuddin,
  Lukasz Kaiser, David Belanger, Lucy Colwell, and Adrian Weller.
\newblock Rethinking attention with performers.
\newblock In {\em International Conference on Learning Representations}, 2021.

\bibitem[CLK{\etalchar{+}}25]{clk+25}
Bo~Chen, Xiaoyu Li, Yekun Ke, Yingyu Liang, Zhenmei Shi, and Zhao Song.
\newblock Exploring the limits of kv cache compression in visual autoregressive
  transformers.
\newblock {\em arXiv preprint arXiv:2503.14881}, 2025.

\bibitem[CLL{\etalchar{+}}25]{cll+25}
Yifang Chen, Xiaoyu Li, Yingyu Liang, Zhenmei Shi, Zhao Song, and Yu~Tian.
\newblock Time and memory trade-off of kv-cache compression in tensor
  transformer decoding.
\newblock {\em arXiv preprint arXiv:2503.11108}, 2025.

\bibitem[CLS{\etalchar{+}}25]{cls+25}
Bo~Chen, Yingyu Liang, Zhizhou Sha, Zhenmei Shi, and Zhao Song.
\newblock {HSR}-enhanced sparse attention acceleration.
\newblock In {\em The Second Conference on Parsimony and Learning (Proceedings
  Track)}, 2025.

\bibitem[CMS{\etalchar{+}}20]{cms+20}
Nicolas Carion, Francisco Massa, Gabriel Synnaeve, Nicolas Usunier, Alexander
  Kirillov, and Sergey Zagoruyko.
\newblock End-to-end object detection with transformers.
\newblock In {\em Proceedings of the European Conference on Computer Vision
  (ECCV)}, pages 213--229. Springer, 2020.

\bibitem[CW17]{cw17}
Kenneth~L Clarkson and David~P Woodruff.
\newblock Low-rank approximation and regression in input sparsity time.
\newblock {\em Journal of the ACM (JACM)}, 63(6):1--45, 2017.

\bibitem[DBK{\etalchar{+}}21]{dbk+21}
Alexey Dosovitskiy, Lucas Beyer, Alexander Kolesnikov, Dirk Weissenborn,
  Xiaohua Zhai, Thomas Unterthiner, Mostafa Dehghani, Matthias Minderer, Georg
  Heigold, Sylvain Gelly, Jakob Uszkoreit, and Neil Houlsby.
\newblock An image is worth 16x16 words: Transformers for image recognition at
  scale.
\newblock In {\em International Conference on Learning Representations}, 2021.

\bibitem[DCLT19]{dclt19}
Jacob Devlin, Ming-Wei Chang, Kenton Lee, and Kristina Toutanova.
\newblock {BERT}: Pre-training of deep bidirectional transformers for language
  understanding.
\newblock In {\em Proceedings of the 2019 Conference of the North American
  Chapter of the Association for Computational Linguistics: Human Language
  Technologies (NAACL-HLT)}, pages 4171--4186. Association for Computational
  Linguistics, 2019.

\bibitem[DH96]{dh96}
Christoph D{\"u}rr and Peter H{\o}yer.
\newblock A quantum algorithm for finding the minimum.
\newblock arXiv:quant-ph/9607014, 1996.

\bibitem[DKM06]{dkm06}
Petros Drineas, Ravi Kannan, and Michael~W. Mahoney.
\newblock Fast {Monte Carlo} algorithms for matrices {I}: Approximating matrix
  multiplication.
\newblock {\em SIAM Journal on Computing}, 36(1):132--157, 2006.

\bibitem[DKOD20]{dkod20}
Giannis Daras, Nikita Kitaev, Augustus Odena, and Alexandros~G. Dimakis.
\newblock Smyrf: Efficient attention using asymmetric clustering.
\newblock In {\em Advances in Neural Information Processing Systems},
  volume~33, pages 6470--6481, 2020.

\bibitem[FLC{\etalchar{+}}25]{flc+25}
Yuan Feng, Junlin Lv, Yukun Cao, Xike Xie, and S.~Kevin Zhou.
\newblock Identify critical kv cache in llm inference from an output
  perturbation perspective.
\newblock {\em arXiv preprint arXiv:2502.03805}, 2025.

\bibitem[Gro96]{g96}
Lov~K. Grover.
\newblock A fast quantum mechanical algorithm for database search.
\newblock In {\em Proceedings of the Twenty-Eighth Annual ACM Symposium on
  Theory of Computing}, STOC '96, page 212–219, New York, NY, USA, 1996.
  Association for Computing Machinery.

\bibitem[GSYZ23]{gsyz23}
Yeqi Gao, Zhao Song, Xin Yang, and Ruizhe Zhang.
\newblock Fast quantum algorithm for attention computation, 2023.

\bibitem[GSYZ24]{gsyz24}
Yuzhou Gu, Zhao Song, Junze Yin, and Lichen Zhang.
\newblock Low rank matrix completion via robust alternating minimization in
  nearly linear time.
\newblock In {\em The Twelfth International Conference on Learning
  Representations (ICLR)}, 2024.

\bibitem[GXL{\etalchar{+}}22]{gxl+22}
Meng-Hao Guo, Tian-Xing Xu, Jiang-Jiang Liu, Zheng-Ning Liu, Peng-Tao Jiang,
  Tai-Jiang Mu, Song-Hai Zhang, Ralph~R Martin, Ming-Ming Cheng, and Shi-Min
  Hu.
\newblock Attention mechanisms in computer vision: A survey.
\newblock {\em Computational Visual Media}, 8(3):331--368, 2022.

\bibitem[Hig02]{h02}
Nicholas~J. Higham.
\newblock {\em Accuracy and Stability of Numerical Algorithms}.
\newblock Society for Industrial and Applied Mathematics, Philadelphia, PA, 2
  edition, 2002.

\bibitem[HJK{\etalchar{+}}24]{hjk+24}
Insu Han, Rajesh Jayaram, Amin Karbasi, Vahab Mirrokni, David~P. Woodruff, and
  Amir Zandieh.
\newblock Hyperattention: Long-context attention in near-linear time.
\newblock In {\em International Conference on Learning Representations (ICLR)},
  2024.

\bibitem[HTF09]{htf09}
Trevor Hastie, Robert Tibshirani, and Jerome Friedman.
\newblock {\em The Elements of Statistical Learning: Data Mining, Inference,
  and Prediction}.
\newblock Springer, New York, 2nd edition, 2009.

\bibitem[HvdH21]{hh21}
David Harvey and Joris van~der Hoeven.
\newblock {Integer multiplication in time O(n log n)}.
\newblock {\em {Annals of Mathematics}}, March 2021.

\bibitem[IKSW25]{iksw25}
Piotr Indyk, Michael Kapralov, Kshiteej Sheth, and Tal Wagner.
\newblock Improved algorithms for kernel matrix-vector multiplication under
  sparsity assumptions.
\newblock In {\em Proceedings of the International Conference on Learning
  Representations (ICLR)}, May 2025.

\bibitem[JYS{\etalchar{+}}20]{jys+20}
Xiaoqi Jiao, Yichun Yin, Lifeng Shang, Xin Jiang, Xiao Chen, Linlin Li, Fang
  Wang, and Qun Liu.
\newblock {TinyBERT}: Distilling {BERT} for natural language understanding.
\newblock In {\em Findings of the Association for Computational Linguistics:
  EMNLP 2020}, pages 4163--4174. Association for Computational Linguistics,
  2020.

\bibitem[KBKW25]{kbkw25}
Ravindran Kannan, Chiranjib Bhattacharyya, Praneeth Kacham, and David~P.
  Woodruff.
\newblock Levattention: Time, space, and streaming efficient algorithm for
  heavy attentions.
\newblock In {\em Proceedings of the International Conference on Learning
  Representations (ICLR)}, 2025.

\bibitem[KKL20]{kkl20}
Nikita Kitaev, Łukasz Kaiser, and Anselm Levskaya.
\newblock Reformer: The efficient transformer.
\newblock In {\em International Conference on Learning Representations}, 2020.

\bibitem[KLLP19]{kllp19}
Iordanis Kerenidis, Jonas Landman, Alessandro Luongo, and Anupam Prakash.
\newblock q-means: A quantum algorithm for unsupervised machine learning.
\newblock In {\em Advances in Neural Information Processing Systems},
  volume~32, pages 4134--4144, 2019.

\bibitem[KLP20]{klp20}
Iordanis Kerenidis, Jonas Landman, and Anupam Prakash.
\newblock Quantum algorithms for deep convolutional neural networks.
\newblock In {\em Proceedings of the 8th International Conference on Learning
  Representations (ICLR)}, 2020.

\bibitem[KMZ24]{kmz24}
Praneeth Kacham, Vahab Mirrokni, and Peilin Zhong.
\newblock Polysketchformer: Fast transformers via sketching polynomial kernels.
\newblock In {\em Proceedings of the 41st International Conference on Machine
  Learning}, 2024.

\bibitem[KP17]{kp17}
Iordanis Kerenidis and Anupam Prakash.
\newblock Quantum recommendation systems.
\newblock In {\em 8th Innovations in Theoretical Computer Science Conference
  (ITCS 2017)}, volume~67 of {\em Leibniz International Proceedings in
  Informatics (LIPIcs)}, pages 49:1--49:21. Schloss Dagstuhl--Leibniz-Zentrum
  f{\"u}r Informatik, 2017.

\bibitem[KP20]{kp20}
Iordanis Kerenidis and Anupam Prakash.
\newblock A quantum interior point method for lps and sdps.
\newblock {\em ACM Transactions on Quantum Computing}, 1(1):1--32, 2020.

\bibitem[KPS21]{kps21}
Iordanis Kerenidis, Anupam Prakash, and D{\'a}niel Szil{\'a}gyi.
\newblock Quantum algorithms for {S}econd-{O}rder {C}one {P}rogramming and
  {S}upport {V}ector {M}achines.
\newblock {\em Quantum}, 5:427, 2021.

\bibitem[KVPF20]{kvpf20}
Angelos Katharopoulos, Apoorv Vyas, Nikolaos Pappas, and François Fleuret.
\newblock Transformers are rnns: Fast autoregressive transformers with linear
  attention.
\newblock In {\em Proceedings of the 37th International Conference on Machine
  Learning}, pages 5156--5165. PMLR, 2020.

\bibitem[KWZ{\etalchar{+}}24]{kwz+24}
Lilly Kumari, Shengjie Wang, Tianyi Zhou, Nikhil Sarda, Anthony Rowe, and Jeff
  Bilmes.
\newblock Bumblebee: Dynamic kv-cache streaming submodular summarization for
  infinite-context transformers.
\newblock In {\em Proceedings of the Conference on Learning for Molecules
  (COLM)}, 2024.

\bibitem[LCW19]{lcw19}
Tongyang Li, Shouvanik Chakrabarti, and Xiaodi Wu.
\newblock Sublinear quantum algorithms for training linear and kernel-based
  classifiers.
\newblock In Kamalika Chaudhuri and Ruslan Salakhutdinov, editors, {\em
  Proceedings of the 36th International Conference on Machine Learning},
  volume~97 of {\em Proceedings of Machine Learning Research}, pages
  3815--3824. PMLR, June 2019.

\bibitem[LEW{\etalchar{+}}22]{lew+22}
Xiaoyun Liu, Daniel Esser, Brandon Wagstaff, Anna Zavodni, Naomi Matsuura,
  Jonathan Kelly, and Eric Diller.
\newblock Capsule robot pose and mechanism state detection in ultrasound using
  attention-based hierarchical deep learning.
\newblock {\em Scientific Reports}, 12(1):21130, 2022.

\bibitem[LFX{\etalchar{+}}24]{lfx+24}
Aixin Liu, Bei Feng, Bing Xue, Bingxuan Wang, Bochao Wu, Chengda Lu, Chenggang
  Zhao, Chengqi Deng, Chenyu Zhang, Chong Ruan, et~al.
\newblock Deepseek-v3 technical report.
\newblock {\em arXiv preprint arXiv:2412.19437}, 2024.

\bibitem[LRO{\etalchar{+}}24]{lro+24}
Minghui Liu, Tahseen Rabbani, Tony O'Halloran, Ananth Sankaralingam, Mary-Anne
  Hartley, Brian Gravelle, Furong Huang, Cornelia Fermüller, and Yiannis
  Aloimonos.
\newblock Hashevict: A pre-attention kv cache eviction strategy using
  locality-sensitive hashing.
\newblock {\em arXiv preprint arXiv:2412.16187}, 2024.

\bibitem[LZ22]{lz22}
Tongyang Li and Ruizhe Zhang.
\newblock Quantum speedups of optimizing approximately convex functions with
  applications to logarithmic regret stochastic convex bandits.
\newblock In {\em Advances in Neural Information Processing Systems},
  volume~35, pages 19565--19577, 2022.

\bibitem[MM17]{mm17}
Cameron Musco and Christopher Musco.
\newblock Recursive sampling for the nystr\"{o}m method.
\newblock In {\em Proceedings of the 31st International Conference on Neural
  Information Processing Systems}, NIPS'17, page 3836–3848, Red Hook, NY,
  USA, 2017. Curran Associates Inc.

\bibitem[MMS18]{mms18}
Cameron Musco, Christopher Musco, and Aaron Sidford.
\newblock Stability of the lanczos method for matrix function approximation.
\newblock In {\em Proceedings of the Twenty-Ninth Annual ACM-SIAM Symposium on
  Discrete Algorithms}, SODA '18, page 1605–1624, USA, 2018. Society for
  Industrial and Applied Mathematics.

\bibitem[{Ope}26]{OpenAIMM26}
{OpenAI}.
\newblock Complex matrix multiplication below {2.258} and rectangular bounds.
\newblock
  \href{https://github.com/openai/math/tree/main/preprints/Complex-Matrix-Multiplication-Below-2.258-and-Rectangular-Bounds-September-24-2026}{OpenAI
  Math Release preprint}, September 2026.

\bibitem[PDC{\etalchar{+}}23]{pdc+23}
Reiner Pope, Sholto Douglas, Aakanksha Chowdhery, Jacob Devlin, James Bradbury,
  Anselm Levskaya, Jonathan Heek, Kefan Xiao, Shivani Agrawal, and Jeff Dean.
\newblock Efficiently scaling transformer inference.
\newblock In {\em Proceedings of the 6th Conference on Machine Learning and
  Systems (MLSys)}, 2023.

\bibitem[PPY{\etalchar{+}}21]{ppy+21}
Hao Peng, Nikolaos Pappas, Dani Yogatama, Roy Schwartz, Noah~A. Smith, and
  Lingpeng Kong.
\newblock Random feature attention.
\newblock In {\em Proceedings of the 9th International Conference on Learning
  Representations (ICLR)}, 2021.

\bibitem[RSR{\etalchar{+}}20]{rsr+20}
Colin Raffel, Noam Shazeer, Adam Roberts, Katherine Lee, Sharan Narang, Michael
  Matena, Yanqi Zhou, Wei Li, and Peter~J. Liu.
\newblock Exploring the limits of transfer learning with a unified text-to-text
  transformer.
\newblock {\em Journal of Machine Learning Research}, 21:1--67, 2020.

\bibitem[RSVG21]{rsvg21}
Aurko Roy, Mohammad Saffar, Ashish Vaswani, and David Grangier.
\newblock Efficient content-based sparse attention with routing transformers.
\newblock {\em Transactions of the Association for Computational Linguistics},
  9:53--68, 2021.

\bibitem[Sha25]{s25_soda}
Rikhav Shah.
\newblock Hermitian diagonalization in linear precision.
\newblock In {\em Proceedings of the 2025 Annual ACM-SIAM Symposium on Discrete
  Algorithms (SODA)}, pages 5599--5615, 2025.

\bibitem[Sob25]{s25_icalp}
Aleksandros Sobczyk.
\newblock Deterministic complexity analysis of hermitian eigenproblems.
\newblock In Keren Censor{-}Hillel, Fabrizio Grandoni, Jo{\"{e}}l Ouaknine, and
  Gabriele Puppis, editors, {\em 52nd International Colloquium on Automata,
  Languages, and Programming, {ICALP} 2025, Aarhus, Denmark, July 8-11, 2025},
  volume 334 of {\em LIPIcs}, pages 131:1--131:21. Schloss Dagstuhl -
  Leibniz-Zentrum f{\"{u}}r Informatik, 2025.

\bibitem[SYY22]{syy22}
Zhiqing Sun, Yiming Yang, and Shinjae Yoo.
\newblock Sparse attention with learning to hash.
\newblock In {\em International Conference on Learning Representations}, 2022.

\bibitem[SYZ24]{syz24}
Zhao Song, Junze Yin, and Lichen Zhang.
\newblock Solving attention kernel regression problem via pre-conditioner.
\newblock In {\em International Conference on Artificial Intelligence and
  Statistics}, pages 208--216. PMLR, 2024.

\bibitem[SZ23]{sz23}
Aaron Sidford and Chenyi Zhang.
\newblock Quantum speedups for stochastic optimization.
\newblock In {\em Advances in Neural Information Processing Systems 36 (NeurIPS
  2023)}, pages 1--12, 2023.

\bibitem[TAB{\etalchar{+}}23]{tab+23}
Gemini Team, Rohan Anil, Sebastian Borgeaud, Jean-Baptiste Alayrac, Jiahui Yu,
  Radu Soricut, Johan Schalkwyk, Andrew~M Dai, Anja Hauth, Katie Millican,
  et~al.
\newblock Gemini: a family of highly capable multimodal models.
\newblock {\em arXiv preprint arXiv:2312.11805}, 2023.

\bibitem[TLI{\etalchar{+}}23]{tli+23}
Hugo Touvron, Thibaut Lavril, Gautier Izacard, Xavier Martinet, Marie-Anne
  Lachaux, Timoth{\'e}e Lacroix, Baptiste Rozi{\`e}re, Naman Goyal, Eric
  Hambro, Faisal Azhar, et~al.
\newblock Llama: Open and efficient foundation language models.
\newblock {\em arXiv preprint arXiv:2302.13971}, 2023.

\bibitem[vAG19]{vg19}
Joran van Apeldoorn and Andr{\'a}s Gily{\'e}n.
\newblock Improvements in quantum sdp-solving with applications.
\newblock In {\em 46th International Colloquium on Automata, Languages, and
  Programming (ICALP 2019)}, volume 132 of {\em Leibniz International
  Proceedings in Informatics (LIPIcs)}, pages 99:1--99:15. Schloss
  Dagstuhl--Leibniz-Zentrum f{\"u}r Informatik, 2019.

\bibitem[vAGGdW20a]{vggw20_convex}
Joran van Apeldoorn, Andr{\'a}s Gily{\'e}n, Sander Gribling, and Ronald
  de~Wolf.
\newblock Convex optimization using quantum oracles.
\newblock {\em Quantum}, 4:220, 2020.

\bibitem[vAGGdW20b]{vggw20_sdp}
Joran van Apeldoorn, Andr{\'a}s Gily{\'e}n, Sander Gribling, and Ronald
  de~Wolf.
\newblock Quantum sdp-solvers: Better upper and lower bounds.
\newblock {\em Quantum}, 4:230, 2020.

\bibitem[vAGL{\etalchar{+}}21]{vgl+21}
Joran van Apeldoorn, Sander Gribling, Yinan Li, Harold Nieuwboer, Michael
  Walter, and Ronald de~Wolf.
\newblock Quantum algorithms for matrix scaling and matrix balancing.
\newblock In {\em 48th International Colloquium on Automata, Languages, and
  Programming (ICALP 2021)}, volume 198 of {\em Leibniz International
  Proceedings in Informatics (LIPIcs)}, pages 110:1--110:17. Schloss Dagstuhl
  -- Leibniz-Zentrum f{\"u}r Informatik, 2021.

\bibitem[vdBSZ24]{bsz24}
Jan van~den Brand, Zhao Song, and Tianyi Zhou.
\newblock Algorithm and hardness for dynamic attention maintenance in large
  language models.
\newblock In {\em International Conference on Machine Learning}, pages
  49008--49028. PMLR, 2024.

\bibitem[VSP{\etalchar{+}}17]{vsp+17}
Ashish Vaswani, Noam Shazeer, Niki Parmar, Jakob Uszkoreit, Llion Jones,
  Aidan~N Gomez, {\L}ukasz Kaiser, and Illia Polosukhin.
\newblock Attention is all you need.
\newblock {\em Advances in neural information processing systems}, 30, 2017.

\bibitem[WDCX21]{wdcx21}
Yuanchao Wang, Wenji Du, Chenghao Cai, and Yanyan Xu.
\newblock Explaining the attention mechanism of end-to-end speech recognition
  using decision trees.
\newblock {\em arXiv preprint arXiv:2110.03879}, 2021.

\bibitem[WLK{\etalchar{+}}20]{wlk+20}
Sinong Wang, Belinda~Z. Li, Madian Khabsa, Han Fang, and Hao Ma.
\newblock Linformer: Self-attention with linear complexity, 2020.

\bibitem[WS00]{ws00}
Christopher K.~I. Williams and Matthias Seeger.
\newblock Using the nystr\"{o}m method to speed up kernel machines.
\newblock In {\em Proceedings of the 14th International Conference on Neural
  Information Processing Systems}, NIPS'00, page 661–667, Cambridge, MA, USA,
  2000. MIT Press.

\bibitem[WSG{\etalchar{+}}25]{wsg+25}
Evan~Pete Walsh, Luca Soldaini, Dirk Groeneveld, Kyle Lo, Shane Arora, Akshita
  Bhagia, Yuling Gu, Shengyi Huang, Matt Jordan, Nathan Lambert, Dustin
  Schwenk, Oyvind Tafjord, Taira Anderson, David Atkinson, Faeze Brahman,
  Christopher Clark, Pradeep Dasigi, Nouha Dziri, Allyson Ettinger, Michal
  Guerquin, David Heineman, Hamish Ivison, Pang~Wei Koh, Jiacheng Liu, Saumya
  Malik, William Merrill, Lester James~Validad Miranda, Jacob Morrison, Tyler
  Murray, Crystal Nam, Jake Poznanski, Valentina Pyatkin, Aman Rangapur,
  Michael Schmitz, Sam Skjonsberg, David Wadden, Christopher Wilhelm, Michael
  Wilson, Luke Zettlemoyer, Ali Farhadi, Noah~A. Smith, and Hannaneh
  Hajishirzi.
\newblock 2 {OLM}o 2 furious ({COLM}{\textquoteright}s version).
\newblock In {\em Second Conference on Language Modeling}, 2025.

\bibitem[WZL24]{wzl24}
Hao Wang, Chenyi Zhang, and Tongyang Li.
\newblock Near-optimal quantum algorithm for minimizing the maximal loss.
\newblock In {\em Proceedings of the 12th International Conference on Learning
  Representations (ICLR)}, 2024.

\bibitem[XCLJ23]{xclj23}
Yecheng Xue, Xiaoyu Chen, Tongyang Li, and Shaofeng H.-C. Jiang.
\newblock Near-optimal quantum coreset construction algorithms for clustering.
\newblock In {\em Proceedings of the 40th International Conference on Machine
  Learning}, volume 202 of {\em Proceedings of Machine Learning Research},
  pages 38881--38912. PMLR, 2023.

\bibitem[XZC{\etalchar{+}}21]{xzc+21}
Yunyang Xiong, Zhanpeng Zeng, Rudrasis Chakraborty, Mingxing Tan, Glenn Fung,
  Yin Li, and Vikas Singh.
\newblock Nyströmformer: A nyström-based algorithm for approximating
  self-attention.
\newblock In {\em AAAI}, pages 14138--14148. AAAI Press, 2021.

\bibitem[YDY{\etalchar{+}}19]{ydy+19}
Zhilin Yang, Zihang Dai, Yiming Yang, Jaime Carbonell, Ruslan Salakhutdinov,
  and Quoc~V. Le.
\newblock Xlnet: Generalized autoregressive pretraining for language
  understanding.
\newblock In {\em Advances in Neural Information Processing Systems},
  volume~32, pages 5754--5764. Curran Associates, Inc., 2019.

\bibitem[ZGD{\etalchar{+}}20]{zgd+20}
Manzil Zaheer, Guru Guruganesh, Kumar~Avinava Dubey, Joshua Ainslie, Chris
  Alberti, Santiago Ontañón, Philip Pham, Anirudh Ravula, Qifan Wang,
  Li~Yang, and Amr Ahmed.
\newblock Big bird: Transformers for longer sequences.
\newblock In {\em Advances in Neural Information Processing Systems},
  volume~33, pages 17283--17297, 2020.

\bibitem[Zha05]{z05}
Tong Zhang.
\newblock Learning bounds for kernel regression using effective data
  dimensionality.
\newblock {\em Neural Computation}, 17(9):2077--2098, 2005.

\bibitem[ZHDK23]{zhdk23}
Amir Zandieh, Insu Han, Majid Daliri, and Amin Karbasi.
\newblock Kdeformer: Accelerating transformers via kernel density estimation.
\newblock In {\em Proceedings of the 40th International Conference on Machine
  Learning}, volume 202 of {\em Proceedings of Machine Learning Research},
  pages 40605--40623. PMLR, 2023.

\bibitem[ZHMK24]{zhmk24}
Amir Zandieh, Insu Han, Vahab Mirrokni, and Amin Karbasi.
\newblock Subgen: Token generation in sublinear time and memory.
\newblock {\em arXiv preprint arXiv:2402.06082}, 2024.

\bibitem[ZWL{\etalchar{+}}24]{zwl+24}
Rongzhi Zhang, Kuang Wang, Liyuan Liu, Shuohang Wang, Hao Cheng, Chao Zhang,
  and Yelong Shen.
\newblock Lorc: Low-rank compression for llms kv cache with a progressive
  compression strategy.
\newblock In {\em NeurIPS 2024 Workshop on Model Compression}, 2024.

\bibitem[ZZF{\etalchar{+}}24]{zzf+24}
Yexin Zhang, Chenyi Zhang, Cong Fang, Liwei Wang, and Tongyang Li.
\newblock Quantum algorithms and lower bounds for finite-sum optimization.
\newblock In {\em Proceedings of the 41st International Conference on Machine
  Learning}, volume 235 of {\em Proceedings of Machine Learning Research},
  pages 12345--12356. PMLR, 2024.

\bibitem[ZZP{\etalchar{+}}21]{zzp+21}
Haoyi Zhou, Shanghang Zhang, Jieqi Peng, Shuai Zhang, Jianxin Li, Hui Xiong,
  and Wancai Zhang.
\newblock Informer: Beyond efficient transformer for long sequence time-series
  forecasting.
\newblock In {\em Proceedings of the AAAI Conference on Artificial
  Intelligence}, volume~35, pages 11106--11115. AAAI Press, 2021.

\end{thebibliography}

\newpage
\onecolumn
\appendix

\begin{center}
    \textbf{\LARGE Appendix }
\end{center}


\ifdefined\isarxiv
\else
\section*{LLM Usage Disclosure}

In the original version of this paper, LLMs were used only to polish language, such as grammar and wording, and did not contribute to idea creation or writing. The corrections in this version were found and written with the help of Claude Opus 5.5. The authors have verified the corrections and take full responsibility for this paper's content.
\fi
\fi

\end{document}